\documentclass[oneside,11pt]{article}

\usepackage[usenames, dvipsnames]{color}
\usepackage{graphicx}
\usepackage{epstopdf}
\usepackage{geometry}
\usepackage[english]{babel}
\usepackage[numbers]{natbib}
\usepackage{amssymb,amsmath,amsfonts,graphicx,epsfig}
\usepackage{authblk}
\usepackage{caption}

\providecommand{\keywords}[1] {
  \small
  \textbf{\textit{Keywords---}} #1
}

\newcommand{\comment}[1]{}

\newtheorem{theorem}{Theorem}
\newtheorem{definition}{Definition}
\newtheorem{prop}{Proposition}
\newcommand{\QED} {\hfill$\square$}
\newtheorem{lem}{Lemma}
\newenvironment{proof} {\par \noindent \textbf{Proof. }}{\QED \par \smallskip \par}

\begin{document}

\title{Another estimation of Laplacian spectrum of the Kronecker product of graphs}
\date{}

\author[1]{Milan Ba\v{s}i\'{c}}
\author[2]{Branko Arsi\'{c}}
\author[3]{Zoran Obradovi\'{c}}
\affil[1]{Department of Computer Science, University of Ni\v{s},
Serbia} \affil[2]{Department of Mathematics and Informatics,
University of Kragujevac, Serbia} \affil[3]{Department of Computer
and Information Sciences, Center for Data Analytics and Biomedical
Informatics, Temple University,
          Philadelphia, PA, USA}
\affil[ ]{\textit{basic\_milan@yahoo.com, brankoarsic@kg.ac.rs,
zoran.obradovic@temple.edu}}

\comment{
\author{\name Milan Ba\v{s}i\'{c} \email basic\_milan@yahoo.com \\
       \addr Department of Computer Science\\
       University of Ni\v{s}\\
       Vi\v{s}egradska 33, 18000 Ni\v{s}, Serbia
       \AND
       \name Branko Arsi\'{c} \email brankoarsic@kg.ac.rs \\
       \addr Department of Mathematics and Informatics\\
       University of Kragujevac\\
       Radoja Domanovi\'{c}a 12, 34000 Kragujevac, Serbia
       \AND
       \name Zoran Obradovi\'{c} \email zoran.obradovic@temple.edu \\
       \addr Department of Computer and Information Sciences\\
       Center for Data Analytics and Biomedical Informatics \\
       Temple University\\
       Philadelphia, PA, USA
} }

\maketitle

\begin{abstract}

The relationships between eigenvalues and eigenvectors of a product
graph and those of its factor graphs have been known for the
standard products, while characterization of Laplacian eigenvalues
and eigenvectors of the Kronecker product of graphs using the
Laplacian spectra and eigenvectors of the factors turned out to be quite challenging
and has remained an open problem to date.
Several approaches for the estimation of Laplacian spectrum of the Kronecker product of graphs have been proposed in recent years.  
However, it turns out that not all the methods are practical to apply in network science models, particularly in the context of multilayer networks. Here we develop a practical and computationally efficient method to estimate Laplacian spectra of this graph product from spectral properties of their factor graphs which is more stable than the alternatives proposed in the literature. 
We emphasize that a median of the percentage errors of our estimated Laplacian spectrum almost coincides with the $x$-axis, unlike the alternatives which have sudden jumps at the beginning followed by a gradual decrease for the
percentage errors. The percentage errors confined (confidence of the estimations) up to $\pm$10\% for all considered approximations, depending on a graph density.
Moreover, we theoretically prove that the percentage errors becomes smaller when the network grows or the edge density level increases.
Additionally, some novel theoretical results considering the exact formulas and lower bounds related to the certain correlation coefficients corresponding to the estimated eigenvectors are presented.

\end{abstract}


\section{Introduction}

Many real-life interactions throughout nature and society, such as
protein-protein interaction networks \cite{prvzulj2011protein},
connections among image pixels \cite{shi2000normalized}, Internet
social networks \cite{arsic2016facebook}, the evolution of a quantum
system \cite{basic2014weighted} etc., could be naturally described
and represented in the context of large networks. However, the
properties of such large networks can not be easily determined
because of a large computational complexity of methods and
algorithms performed on their corresponding graph matrices. Fortunately, large networks are often composed of
several smaller pieces, for example motifs \cite{milo2002network},
communities \cite{girvan2002community}, or layers
\cite{de2013mathematical}. In this case, by using the properties of
these smaller structures, we can determine the properties of large
networks obtained by using some operations \cite{arenas2007size,
skardal2012hierarchical}. In graph theory there are three
fundamental graph products which refer to the large network's
construction from two or more small graphs: Cartesian product,
Kronecker (direct) product, and strong product. In each case, the
product of graphs $G$ and $H$ is a graph whose vertex set is the
Cartesian product $V(G) \times V (H)$ of sets, while each product
has different rules for edge creation. Computer science is one of
the many fields (such as mathematics and engineering) in which graph products, with their own set of
applications and theoretical interpretations, are becoming
commonplace. As one specific example, large networks such as the
Internet graph, with several hundred million hosts, can be
efficiently modeled by subgraphs of powers of small graphs with
respect to the Kronecker product \cite{leskovec2010kronecker}. More
recently, graph products have also began to appear in
network science, where multiplication of graphs are often used as a
formal way to describe certain types of multilayer network
topologies
\cite{de2013mathematical}\cite{sole2013spectral}\cite{kivela2014multilayer}.
Products of graphs that make use of spectral methods have also found
important applications in interconnection networks, massively
parallel computer architectures and diffusion schemes
\cite{elsasser2004optimal}.

It was recognized in about the last twenty years that graph spectra
have many important applications in various areas, especially in the
fields of computer sciences (see, e.g.,
\cite{cvetkovic2011selected}\cite{cvetkovic2011graph}), such as
Internet technologies, computer vision, pattern recognition, data
mining, multiprocessor systems, statistical databases and many
others. One of the important questions to be addressed in this area,
and which have been studied extensively by many researchers, is how
to characterize spectral properties of a product graph using those
of its factor graphs. Relationships between spectral properties of a
product graph and those of its factor graphs have been known for the
spectra of degree and adjacency matrices for all of the three
products, as well as the Laplacian spectra for Cartesian product
\cite{sayama2016estimation}. Results describing the adjacency matrix
and its spectra of the product graphs can be also found in
\cite{brouwer2011spectra} and \cite{cvetkovic1980spectra}, while a
complete characterization of the Laplacian spectrum of the Cartesian
product of two graphs has been done by Merris
\cite{merris1994laplacian}. In the paper \cite{glass2017structured},
the authors tried to exploit the benefits of the Kronecker graph
representation, which is used as a replacement for the multilayer
network. However, they had to face an open problem, because the
Laplacian spectrum of the Kronecker product of two graphs graphs can
not be characterized by using the Laplacian spectra of the factors.
In \cite{barik2015laplacian}, the authors gave the explicit complete
characterization of the Laplacian spectrum of the Kronecker product
of two graphs in some particular cases. Since it seems that an
explicit formula can not be obtained for the general case, in
\cite{sayama2016estimation} the authors developed empirical methods
to estimate the Laplacian spectra of the Kronecker of graphs from
spectral properties of their factor graphs.

In this paper we develop an alternative practical method for an
estimation of the the Laplacian spectrum and eigenvectors of the Kronecker
(direct) product of two graphs. We noticed that estimated
eigenvalues and eigenvectors of these approximations express
different behavior depending on the type of network topology. The
effectiveness of the proposed methods are evaluated through
numerical experiments, where experiments are performed on three
types of graphs: Erd\H{o}s-R\'{e}nyi, Barab\'{a}si-Albert and
Watts-Strogatz, while the edge density percentage is varied over
10\%, 30\%, and 65\%. In order to see whether, our novel
approximation or the one proposed by Sayama in \cite{sayama2016estimation}, is
more suitable for the original eigenvalues and eigenvectors, we
compare them in the following two ways. 
First, we give an empirical and some theoretical evidence that the Kronecker product of
eigenvectors of normalized Laplacian matrices of factor graphs
 can be also used as an approximation for the eigenvectors of Laplacian matrix of Kronecker product of graphs.
It can be done by comparing the
correlation coefficients that correspond to the approximated vectors for both approximation in regard to different types and edge density
levels of graphs. 
Then, in order to test how close the estimated
to the original eigenvalues of Laplacian of the Kronecker product of
graphs for both approximations are, the difference between them in terms
of a distribution of percentage errors is reported.
We show that a distribution of percentage
errors between novel estimated and original spectra is more stable
than the error obtained for the Sayama's spectrum and it is almost
uniformly distributed around 0, all in the case of
Erd\H{o}s-R\'{e}nyi and Watts-Strogatz random networks.
It is also noticed that
both approximations produced reasonable estimations of Laplacian
spectra with percentage errors confined within a $\pm$10\% range for
most eigenvalues, with a small variations depending on the type and
edge density levels of random networks.
Moreover, we theoretically prove
that the percentage errors become smaller when the
network grows or the edge density level increases for Erd\H{o}s-R\'{e}nyi random networks.
In the case
of Barab\'{a}si-Albert random networks, similar number of jumps in the graphs of 
percentage errors distribution is noticed for the both proposed estimated spectra.

The remainder of our paper is organized through the following
sections. In Section 2 we will explain the motivation and
assumptions for our alternative approach developed for the
estimation of the Laplacian eigenvalues and eigenvectors of the
Kronecker product of graphs. Moreover, in subsection 2.1 we recall some results and
techniques used in \cite{sayama2016estimation} and provide a proof
that all estimated eigenvalues proposed by Sayama are nonnegative. 
In subsection 2.2 we introduce the Kronecker product of
eigenvectors of normalized Laplacian matrices of factor graphs as a potential approximation for the actual eigenvectors
the Laplacian matrix of Kronecker product of graphs 
and by using them we get the formula (\ref{eq:novel_spectrum2}) for estimating the Laplacian spectra of Kronecker product of graphs.
In Section 3 we
report a behavior of the estimated eigenvalues and eigenvectors (for both approximations)
compared to the original ones with regard
to the different types of graphs and different edge density levels. The comparison between estimated and
original spectra has been done by calculating the percentage error, while the correlation coefficients are used to express the difference between eigenvectors. In subsection 3.1.2 we provide some new theoretical results related to
the correlation coefficients that correspond to the estimated vectors for both approximation
and give certain explanation why the Kronecker product of
eigenvectors of normalized Laplacian matrices of factor graphs can be used as suitable approximation for the actual eigenvectors
the Laplacian matrix of Kronecker product of graphs.
In Theorem ~\ref{thm:coeff cor calculation} and Theorem~\ref{thm:corrcoeff_expected}
we provide exact formulas for the certain correlation coefficients and the expected values of the correlation coefficients, respectively, 
corresponding to the eigenvectors proposed in
\cite{sayama2016estimation}. 
From the formulas for the correlation coefficients follow that they depend only on the degrees of one of the factor graphs  
and hence they are mutually equal.
According to the expected value of the previous correlation
coefficients obtained by Theorem~\ref{thm:corrcoeff_expected} and the inequality given by Theorem~\ref{thm:second inequality}, we obtain that the correlation coefficients corresponding to our estimated vectors in
some cases can be greater than the coefficient correlations corresponding to the eigenvectors proposed in
\cite{sayama2016estimation}. 
Finally, using Theorem~\ref{thm:er_theory}
we give a theoretical explanation of why the estimated eigenvalues
for the random graphs become more accurate to the real values when
the network grows or the edge density level increases. The paper
concludes with a summary of key points and directions for further
work. We also point out that these approximations could have a very important application in learning models based on multilayer networks. \cite{glass2017structured}.




\section{Proposed methods}
\label{sec:gcrf-msn}

Before describing the proposed methods, we provide definitions for
concepts used throughout the paper. By $G = (V_{G},
E_{G})$ we denote a simple connected graph (without loops and multiple edges),
where $V_{G}$ is the set of vertices and $E_{G}\subseteq
{V_{G}\choose 2}$ is a set of edges of $G$. The adjacency matrix $A$
for a graph $G$ with $N$ vertices is an $N\times N$ matrix whose
$(i,j)$ entry is 1 if the $i$-th and $j$-th vertices are adjacent,
and 0 if they are not. A number of the vertices $N$ of a graph $G$
is called the order of a graph $G$. A vertex and an edge are called
incident, if the vertex is one of the two vertices that the edge
connects. The Laplacian matrix of the adjacency matrix $A$ is
defined as $L = D - A$ where $D$ is the degree matrix of $A$ (degree
matrix is a diagonal matrix where each entry $(i,i)$ is equal to the
number of edges incident to $i$-th vertex). The normalized Laplacian
matrix is defined as $\mathcal{L} =
D^{-\frac{1}{2}}LD^{-\frac{1}{2}} = I -
D^{-\frac{1}{2}}AD^{-\frac{1}{2}}$. Let $G = (V_{G}, E_{G})$ and $H
= (V_{H}, E_{H})$ be two simple connected graphs.

\begin{definition}
The Kronecker product of graphs denoted by $G\otimes H$ is a graph
defined on the set of vertices $V_{G}\times V_{H}$ such that two
vertices $(g,h)$ and $(g',h')$ are adjacent if and only if $(g,g')
\in E_{G}$ and $(h, h')\in E_{H}$.
\end{definition}

The Kronecker product of an $N\times N$ matrix $A$ and a $M\times M$
matrix $B$ is the $(NM)\times(NM)$ matrix $A\otimes B$ with elements
defined by $(A\otimes B)_{I,J} = A_{i,j}B_{k,l}$ with $I = M(i-1)+k$
and $J = M(j-1) + l$. \vspace{0.3cm}

In the rest of this section we discuss the spectral decomposition of
the Laplacian of the Kronecker product of graphs from those of its
factor graphs. Because it seems that such an explicit formula does not
exist, we need to apply some approximations in order to
obtain the estimated eigenvalues and eigenvectors.

\subsection{Estimation of Laplacian spectrum of Kronecker product graph by using the Kronecker product of Laplacian eigenvectors of factor graphs}
\label{subsec:sayama_approx}

In the following section we will explain the motivation and
assumptions from \cite{sayama2016estimation} for the proposed
approximation and show some of their properties. The Laplacian of the
Kronecker product of graphs is given by the following
\begin{equation}
\notag
\begin{split}
L_{S_{1}\otimes S_{2}} & = D_{S_{1}\otimes S_{2}} - A_{S_{1}\otimes S_{2}}\\
& = (D_{S_{1}}\otimes D_{S_{2}}) - (A_{S_{1}}\otimes A_{S_{2}})\\
& = D_{S_{1}}\otimes D_{S_{2}} - (D_{S_{1}}-L_{S_{1}})\otimes(D_{S_{2}}-L_{S_{2}})\\
& = L_{S_{1}}\otimes D_{S_{2}} + D_{S_{1}}\otimes L_{S_{2}} -
L_{S_{1}}\otimes L_{S_{2}},
\end{split}
\end{equation}
where $A_{S_{1}}$ and $A_{S_{2}}$ are the adjacency matrices and
$D_{S_{1}}$ and $D_{S_{2}}$ are the degree matrices of graphs
$S_{1}$ and $S_{2}$, respectively, where $|S_{1}| = n_{1}$ and
$|S_{2}|=n_{2}$. The idea of the proposed approximation is to assume
that $w_{i}^{S_{1}}\otimes w_{j}^{S_{2}}$, where $w_{i}^{S_{1}}$ and
$w_{j}^{S_{2}}$ are arbitrary eigenvectors of $L_{S_{1}}$ and
$L_{S_{2}}$ respectively, could be used as a substitute for the true
eigenvectors of $L_{S_{1}\otimes S_{2}}$. A motivation for this
assumption came from the fact that the Laplacian spectra of the
Kronecker product of graphs resemble those of the Cartesian product of
graphs when either factor graph is regular
\cite{barik2015laplacian}. Let $W_{S_{1}}$ and $W_{S_{2}}$ be
$n_{1}\times n_{1}$ and $n_{2}\times n_{2}$ square matrices that
contain all $w_{i}^{S_{1}}$ and $w_{j}^{S_{2}}$ as column vectors,
respectively. By making (mathematically incorrect)
assumption that $D_{S_{1}}W_{S_{1}}\approx W_{S_{1}}D_{S_{1}}$ and
$D_{S_{2}}W_{S_{2}}\approx W_{S_{2}}D_{S_{2}}$ it can be obtained that

\begin{equation}
\begin{split}
L_{S_{1}\otimes S_{2}}(W_{S_{1}}\otimes W_{S_{2}}) & = L_{S_{1}}W_{S_{1}}\otimes D_{S_{2}}W_{S_{2}} + D_{S_{1}}W_{S_{1}}\otimes L_{S_{2}}W_{S_{2}} - L_{S_{1}}W_{S_{1}}\otimes L_{S_{2}}W_{S_{2}}\\
& \approx W_{S_{1}}\Lambda_{S_{1}}\otimes W_{S_{2}}D_{S_{2}} + W_{S_{1}}D_{S_{1}}\otimes W_{S_{2}}\Lambda_{S_{2}} - W_{S_{1}}\Lambda_{S_{1}}\otimes W_{S_{2}}\Lambda_{S_{2}}\\
& = (W_{S_{1}}\otimes W_{S_{2}})\Big(\Lambda_{S_{1}}\otimes
D_{S_{2}} + D_{S_{1}}\otimes \Lambda_{S_{2}} -
\Lambda_{S_{1}}\otimes \Lambda_{S_{2}}\Big),
\end{split}
\label{eq:laplacian_spectrum}
\end{equation}
where $\Lambda_{S_{1}}$ and $\Lambda_{S_{2}}$ are diagonal matrices
with eigenvalues $\mu^{S_{1}}_{i}$ of $L_{S_{1}}$ and
$\mu^{S_{2}}_{j}$ of $L_{S_{2}}$, respectively. From the last
equation, estimated Laplacian spectrum of $S_{1}\otimes S_{2}$ could
be calculated as

\begin{equation}
\mu_{ij} = \{\mu_{i}^{S_{1}}d_{j}^{S_{2}} +
d_{i}^{S_{1}}\mu_{j}^{S_{2}} - \mu_{i}^{S_{1}}\mu_{j}^{S_{2}}\}.
\label{eq:sayama_spectrum}
\end{equation}
where $d_{i}^{S_{1}}$ and $d_{j}^{S_{2}}$ are the diagonal entries
of the degree matrices $D_{S_{1}}$ and $D_{S_{2}}$, respectively.

Here we note that the orderings of $w_{i}^{S_{1}}$ and
$w_{j}^{S_{2}}$ (and hence $\mu_{i}^{S_{1}}$ and $\mu_{j}^{S_{2}}$)
are independent of the vertex orderings in $D_{S_{1}}$ and
$D_{S_{2}}$, respectively. This can help in reducing the
mathematical inaccuracy arising from the mentioned incorrect
assumptions by finding optimal column permutations of $W_{S_{1}}$
and $W_{S_{2}}$ (influencing $\Lambda_{S_{1}}$ and
$\Lambda_{S_{2}}$). Therefore, several types of ordering of
eigenvalues ($\mu_{i}^{S_{1}}$ and $\mu_{j}^{S_{2}}$) of factor
graphs were tested \cite{sayama2016estimation}, while the degree
sequences are fixed in ascending order. It was obtained that the
most effective heuristic method is when the eigenvalues are sorted
in ascending order.

From \eqref{eq:sayama_spectrum} it can be easily seen that the
estimated spectrum always has an eigenvalue of 0, because if
$\mu_{i}^{S_{1}} = 0$ and $\mu_{j}^{S_{2}} = 0$, then $\mu_{ij}=0$.
However, it is not commented in \cite{sayama2016estimation} whether
all other estimated eigenvalues $\mu_{ij}$ are greater than or equal to
0. Notice that ~\eqref{eq:sayama_spectrum} can be rewritten as
follows:

\begin{equation}
\notag \mu_{i}^{S_{1}}(d_{j}^{S_{2}} - \frac{\mu_{j}^{S_{2}}}{2}) +
\mu_{j}^{S_{2}}(d_{i}^{S_{1}} - \frac{\mu_{i}^{S_{1}}}{2})
\,\,\,\,\text{for}\,\, 1 \leq i \leq n_{1},\, 1 \leq j \leq n_{2}.
\end{equation}

If a graph is regular then the absolute values of the eigenvalues of its
adjacency matrix are less than or equal to the regularity of the graph
(according to the Perron-Frobenius theorem, see
\cite{godsil2001algebraic}, pp. 178) and it is clear from the
definition of the Laplacian matrix that all Laplacian eigenvalues are
less than or equal to the double value of the regularity. This implies
that in the case when $S_{1}$ and $S_{2}$ are regular, we have that
$d_{j}^{S_{2}} \geq \frac{\mu_{j}^{S_{2}}}{2}$ and $d_{i}^{S_{1}}
\geq \frac{\mu_{i}^{S_{1}}}{2}$, and therefore $\mu_{ij}\geq 0$. In
the following we prove that these eigenvalues are nonnegative in the
general case.

%
%

By applying Gershgorin circle theorem on Laplacian matrix we can obtain only the inequality $d_{n_{1}}^{S_{1}} -
\frac{\mu_{n_{1}}^{S_{1}}}{2}\geq 0$ (or equivalently $\mu_{n_{1}}^{S_{1}}\leq 2d_{n_{1}}^{S_{1}}$).
Indeed, as every eigenvalue of the $n_1\times n_1$ Laplacian matrix $L=(l_{i,j})_{1\leq i,j \leq n_1}$ 
lies within the union of disks centered at $l_{i,i}=d_i^{S_1}$ with radius $R_{i}=d_{i}^{S_1}$ ($R_{i}$ is the sum of the absolute values of the non-diagonal entries in the $i$-th row for $1\leq i\leq n_1$), we can not conclude that every eigenvalue $\mu_{i}^{S_1}$ lies in the circle 
 centered at $d_i^{S_1}$ with radius $d_{i}^{S_1}$, i.e. $\mu_{i}^{S_1}\leq 2d_i^{S_1}$, $1\leq i\leq n_1-1$ (see
Figure~\ref{fig:gershgorin}) .


It turns out that the inequality $\mu_{i}^{S_1}\leq 2d_i^{S_1}$ can be proved by
using Courant-Fischer theorem for every index
$i$. Namely, it is easy to see that the quadratic form $x^T L x$ in
respect to the Laplace matrix $L$ and an arbitrary vector
$x=(x_1,x_2,\ldots, x_{n_1})$ can be rewritten in the following way
$$
x^T L x = \sum_{j=1}^{n_1} d_j^{S_1}x_j^2-2\sum_{(i,j)\in E(S_1)}
x_ix_j.
$$
Now, using the arithmetic-geometric mean inequality between $x_{i}$
and $x_{j}$, $|2x_ix_j|\leq x_i^2+x_j^2$, it holds that
$-2\sum_{(i,j)\in E(S_1)} x_ix_j\leq\sum_{j=1}^{n_1} d_j^{S_1}x_j^2$
and therefore $x^T L x \leq 2 \sum_{j=1}^{n_1} d_j^{S_1}x_j^2$.
Furthermore, considering $x\in R^i\times \{0\}^{n-i}\subseteq R^n$,
we have in this case that $x^T L x \leq 2 \sum_{j=1}^{i}
d_j^{S_1}x_j^2\leq 2d_i^{S_1}\|x\|^2$. Finally, according to
Courant-Fischer we have that $\mu_i^{S_1}\leq \max_{x\in R^i\times
\{0\}^{n-i}} \frac{x^T L x}{\|x\|^2}\leq
\frac{2d_{i}^{S_{1}}\|x\|^2}{\|x\|^2}=2d_i^{S_1}$ (we have already
mentioned that the degree sequence is set in ascending order, that
is $d_1^{S_1}\leq \ldots\leq d_{n_1}^{S_1}$).



\begin{figure}[h!]
\centering
\includegraphics[width=0.5\textwidth]{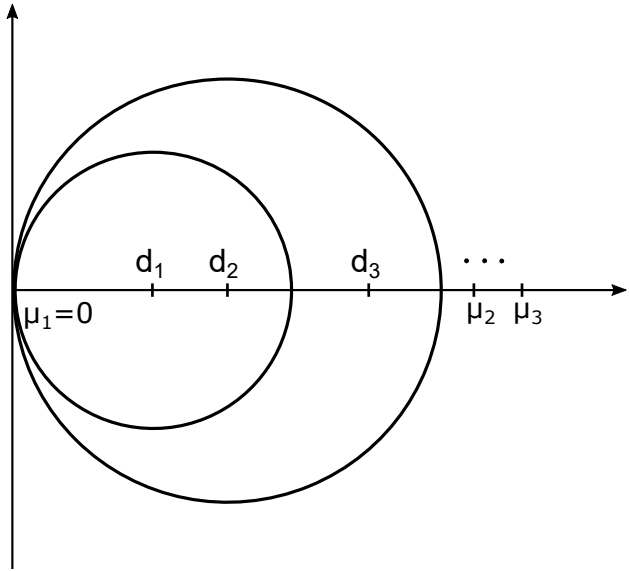}
\caption{Gershgorin disks for Laplacian matrix}
\label{fig:gershgorin}
\end{figure}

The approximations from \cite{sayama2016estimation} are not derived
from rigorous mathematical proofs, but from empirical evidence and
good behavior of estimated eigenvalues and eigenvectors has been
noticed for some types of random graphs. In the following subsection
we propose an estimation of Laplacian spectral decomposition for the
Kronecker product of graphs by using the normalized Laplacian
eigenvectors of factor graphs. We show some differences side by side
(both experimentally and analytically) between these approximations
through the eigenvectors and eigenvalues analysis separately.

\subsection{Estimation of Laplacian spectrum of Kronecker product graph by using the Kronecker product of normalized Laplacian eigenvectors of factor graphs}
\label{subsec:novel_approx}

In this section we propose an alternative approach for estimating the
Laplacian spectrum of the Kronecker product of graphs. The idea
comes from the fact that the normalized Laplacian matrix of the
Kronecker product of graphs can be represented in terms of
normalized Laplacian matrices of factor graphs. Moreover, in some
cases the Kronecker product of the eigenvectors of
$\mathcal{L}_{S_{1}}$ and $\mathcal{L}_{S_{2}}$ gives better
approximation for eigenvectors of $L_{S_{1}\otimes S_{2}}$ than the
Kronecker product of the eigenvectors of $L_{S_{1}}$ and $L_{S_{2}}$. 
Now, we will explain the motivation and assumptions for this novel approach in more detail.

By the definition of the normalized Laplacian and the property
$(A\otimes B)^{-1} = A^{-1}\otimes B^{-1}$, the normalized Laplacian
of the matrix $S_{1}\otimes S_{2}$ can be written in the following
way

\begin{equation}
\notag \mathcal{L}_{S_{1}\otimes S_{2}} = I_{n_{1}}\otimes I_{n_{2}}
- (D_{S_{1}}^{-\frac{1}{2}}\otimes
D_{S_{2}}^{-\frac{1}{2}})(S_{1}\otimes S_{2})
(D_{S_{1}}^{-\frac{1}{2}}\otimes D_{S_{2}}^{-\frac{1}{2}}).
\end{equation}
Using the property of the Kronecker product of matrices, $(A\otimes
B)(C\otimes D) = AC \otimes BD$, we further obtain:
\begin{equation}
\notag \mathcal{L}_{S_{1}\otimes S_{2}} = I_{n_{1}}\otimes I_{n_{2}}
- (D_{S_{1}}^{-\frac{1}{2}} S_{1} D_{S_{1}}^{-\frac{1}{2}})\otimes
(D_{S_{2}}^{-\frac{1}{2}} S_{2} D_{S_{2}}^{-\frac{1}{2}}) =
I_{n_{1}}\otimes I_{n_{2}} - (I_{n_{1}} -
\mathcal{L}_{S_{1}})\otimes (I_{n_{2}} - \mathcal{L}_{S_{2}}).
\end{equation}

Let $\{\lambda_{i}^{S_{1}}\}$ and $\{\lambda_{j}^{S_{2}}\}$ be the
eigenvalues of the matrices $\mathcal{L}_{S_{1}}$ and
$\mathcal{L}_{S_{2}}$, with the corresponding orthonormal
eigenvectors $\{v_{i}^{S_{1}}\}$ and $\{v_{j}^{S_{2}}\}$, where $i =
1, 2, \ldots, n_{1}$ and $j = 1, 2, \ldots, n_{2}$. Denote by
$\Lambda_{S_{1}}$ and $\Lambda_{S_{2}}$ the diagonal matrices whose
diagonal elements are the values $1-\lambda_{i}^{S_{1}}$ and
$1-\lambda_{j}^{S_{2}}$, respectively. Also, $V_{S_{1}}$ and
$V_{S_{2}}$ stand for the square matrices which contain
$v_{i}^{S_{1}}$ and $v_{j}^{S_{2}}$ as column vectors. Using the
spectral decomposition of the matrix $(I_{n_{1}} -
\mathcal{L}_{S_{1}})\otimes (I_{n_{2}} - \mathcal{L}_{S_{2}})$, from
the above equation it follows that

\begin{equation}
\begin{split}
\mathcal{L}_{S_{1}\otimes S_{2}} & = I_{n_{1}}\otimes I_{n_{2}} - (V_{S_{1}}\Lambda_{S_{1}}V_{S_{1}}^{T})\otimes (V_{S_{2}}\Lambda_{S_{2}}V_{S_{2}}^{T})\\
& = I_{n_{1}}\otimes I_{n_{2}} - (V_{S_{1}}\otimes V_{S_{2}})(\Lambda_{S_{1}}\otimes \Lambda_{S_{2}})(V_{S_{1}}\otimes V_{S_{2}})^{T}\\
& = (V_{S_{1}}\otimes V_{S_{2}})(I_{n_{1}}\otimes I_{n_{2}} -
\Lambda_{S_{1}}\otimes \Lambda_{S_{2}})(V_{S_{1}}\otimes
V_{S_{2}})^{T},
\end{split}
\label{eq:spectral_decomposition}
\end{equation}
since $(V_{S_{1}}\otimes V_{S_{2}})(V_{S_{1}}\otimes V_{S_{2}})^{T}
= I_{n_{1}n_{2}}$. This further implies that the normalized
Laplacian matrix of the Kronecker product of graphs have $\{1 -
(1-\lambda_{i}^{S_{1}})(1-\lambda_{j}^{S_{2}})\}$ as eigenvalues and
$\{v_{i}^{S_{1}}\otimes v_{j}^{S_{2}}\}$ as eigenvectors.

Now, put $\Lambda = I_{n_{1}}\otimes I_{n_{2}} -
\Lambda_{S_{1}}\otimes \Lambda_{S_{2}}$ and $D=D_{S_1}\otimes
D_{S_2}$. It is well known that the normalized Laplacian can be
expressed in term of Laplacian matrix as $\mathcal{L} =
D^{-\frac{1}{2}}L D^{-\frac{1}{2}}$. Furthermore, as
$L_{S_{1}\otimes S_{2}}(V_{S_{1}}\otimes V_{S_{2}}) =
D^{\frac{1}{2}} \mathcal{L}_{S_{1}\otimes S_{2}} D^{\frac{1}{2}}
(V_{S_{1}}\otimes V_{S_{2}})$, using the similar assumption like in
the previous subsection that $D_{S_{1}}^{\frac{1}{2}}
V_{S_{1}}\approx V_{S_{1}}D_{S_{1}}^{\frac{1}{2}}$ and
$D_{S_{2}}^{\frac{1}{2}} V_{S_{2}}\approx
V_{S_{2}}D_{S_{2}}^{\frac{1}{2}}$, from
\eqref{eq:spectral_decomposition}, we derive
\begin{equation}
\notag L_{S_{1}\otimes S_{2}}(V_{S_{1}}\otimes V_{S_{2}}) \approx
D^{\frac{1}{2}} \mathcal{L}_{S_{1}\otimes S_{2}} (V_{S_{1}}\otimes
V_{S_{2}}) D^{\frac{1}{2}} = D^{\frac{1}{2}} \Lambda
(V_{S_{1}}\otimes V_{S_{2}}) D^{\frac{1}{2}}.
\end{equation}
Finally, applying the same assumption again we have the following
formula
\begin{equation}
L_{S_{1}\otimes S_{2}}(V_{S_{1}}\otimes V_{S_{2}}) \approx
(D\Lambda)(V_{S_{1}}\otimes V_{S_{2}}). \label{eq:novel_spectrum}
\end{equation}
Inside the first pair of parenthesis of the right-hand side of
\eqref{eq:novel_spectrum} is the diagonal matrix $D\Lambda$ which leads us to a
potential formula for estimating the Laplacian spectrum of the
Kronecker product of graphs, while for the corresponding
eigenvectors we could use eigenvectors of the normalized Laplacian
matrix of the Kronecker product of graphs. Therefore, a
potential formula for estimating the Laplacian spectra of the
Kronecker product of graphs

\begin{equation}
\begin{split}
\mu_{ij} & = \{(1-(1-\lambda_{i}^{S_{1}})(1-\lambda_{j}^{S_{2}}))d_{i}^{S_{1}}d_{j}^{S_{2}}\} \\
& = \{(\lambda_{i}^{S_{1}} + \lambda_{j}^{S_{2}} -
\lambda_{i}^{S_{1}}\lambda_{j}^{S_{2}})d_{i}^{S_{1}}d_{j}^{S_{2}}\},
\end{split}
\label{eq:novel_spectrum2}
\end{equation}

\noindent which are obviously nonnegative. Moreover, the first
eigenvalue is always matched at 0 in both actual and estimated
spectra, because \eqref{eq:novel_spectrum2} guarantees this.

Similarly as in \cite{sayama2016estimation} this approximation shares the property that
the orderings of $v_{i}^{S_{1}}$ and $v_{j}^{S_{2}}$ in $V_{S_{1}}$
and $V_{S_{2}}$ (and hence $\lambda_{i}^{S_{1}}$ and
$\lambda_{j}^{S_{2}}$) are independent of vertex orderings in
$D_{S_{1}}$ and  $D_{S_{2}}$, respectively, and it would be
impractical to try to find true optimal orderings. The following
five heuristic methods that use only degrees and eigenvalues of
factor graphs were tested: \textit{uncorrelated ordering},
\textit{correlated ordering}, \textit{correlated ordering with
randomization}, \textit{anti-correlated ordering} and
\textit{anti-correlated ordering with randomization}. In each
method, it is assumed that the degree sequences ($d_{i}^{S_{1}}$ and
$d_{j}^{S_{2}}$) are already sorted in ascending order, while the
orders of eigenvalues ($\lambda_{i}^{S_{1}}$ and
$\lambda_{j}^{S_{2}}$) are altered differently. The most effective
ordering methods turned out to be correlated ordering
($\lambda_{i}^{S_{1}}$ and $\lambda_{j}^{S_{2}}$ are sorted in
ascending order), as it was obtained for approximation spectrum \cite{sayama2016estimation}.

\section{Estimated eigenvalues and eigenvectors evaluation}
\label{sec:spectra_experiments} 

In this section we report a behavior
of the estimated eigenvalues and eigenvectors, from the presented
approximations, compared to the original
ones with regard to the different types of graphs and different edge
density levels. With these experiments we aim to address the
following:

\begin{itemize}
\item We will show how close the estimated to the original eigenvectors of Laplacian of the Kronecker product of graphs are for these approximations. In order to do that we measure the distribution of vector correlation coefficients between $v_{i}^{S_{1}}\otimes v_{j}^{S_{2}}$ and $L_{S_{1}\otimes S_{2}} (v_{i}^{S_{1}}\otimes v_{j}^{S_{2}})$ as it was done for the eigenvectors $w_{i}^{S_{1}}\otimes w_{j}^{S_{2}}$ in \cite{sayama2016estimation}. In the rest of the section, we give an empirical and some theoretical evidence that the eigenvectors $v_{i}^{S_{1}}\otimes v_{j}^{S_{2}}$ can be also used as an approximation for the eigenvectors of $L_{S_{1}\otimes S_{2}}$.

\item We will show how close estimated to the original eigenvalues of Laplacian of the Kronecker product of graphs are for both approximations. Based on the corresponding estimated spectra~\eqref{eq:sayama_spectrum} and~\eqref{eq:novel_spectrum2}, the difference between estimated and original spectra is reported in terms of a distribution of percentage errors between them. Both approximations produced reasonable estimations of Laplacian spectra with percentage errors confined within a $\pm$10\% range for most eigenvalues.  This error value holds for the sparse graphs. It can be noticed that this error is even smaller for the denser graphs, i. e. about $\pm$5\% and $\pm$2\% when the edge density percentages are 30\% and 65\%, respectively. We also noticed that the median of the percentage errors of our estimated Laplacian spectrum are more stable than in the case of spectrum proposed by Sayama. 
Moreover, we give a theoretical explanation of why the percentage errors of the approximated eigenvalues that correspond to $v_i^{S_1}\otimes v_j^{S_2}$ for the
random graphs become more accurate to the real expected values when
the network grows or the edge density level increases.

\end{itemize}

Experiments are performed on three types of graphs:
Erd\H{o}s-R\'enyi, Barab\'asi-Albert and Watts-Strogatz, while the
edge density percentage is varied over 10\%, 30\%, and 65\%. For the
orders of graphs $G$ and $H$ denoted by $n_{1}$ and $n_{2}$,
respectively, we conduct experiments three times depending on the
orders of graphs $(n_{1}, n_{2})\in \{(30, 50),\- (50, 100),\- (100,
200)\}$.

\subsection{Erd\H{o}s-R\'enyi and Watts-Strogatz graphs}
\label{subsec:er_ws_vectors} Here we describe the behavior of
estimated eigenvectors and eigenvalues for both classes of graphs,
Erd\H{o}s-R\'enyi and Watts-Strogatz, since their vector correlation
coefficients and distributions of percentage errors of the estimated
eigenvalues behave similarly for the same experimental setup. We
also noted a bit smaller errors in the case of Watts-Strogatz than
for Erd\H{o}s-R\'enyi random networks. For both types of graphs we
find that the distribution of correlation coefficients between the vectors $v_{i}^{S_{1}}\otimes
v_{j}^{S_{2}}$  and $L_{S_{1}\otimes S_{2}} (v_{i}^{S_{1}}\otimes
v_{j}^{S_{2}})$, and the vectors $w_{i}^{S_{1}}\otimes
w_{j}^{S_{2}}$ and $L_{S_{1}\otimes S_{2}} (w_{i}^{S_{1}}\otimes
w_{j}^{S_{2}})$  behave very similar to the corresponding values of
the eigenvectors $v_{i}^{S_{1}}\otimes v_{j}^{S_{2}}$ and
$w_{i}^{S_{1}}\otimes w_{j}^{S_{2}}$, when the edge density grows. Further in the paper by a term the correlation coefficients corresponding to the arbitrary eigenvectors $x^{S}_{i}$ of the graph $S$, we will mean the correlation coefficients between the vectors $x^{S}_{i}$ and $L_{S}(x^{S}_{i})$.
Also, we noticed that the shape of the percentage error distribution
across these two network topologies is more consistent (without
sudden jumps) for the estimated spectrum corresponding to the
eigenvectors $v_{i}^{S_{1}}\otimes v_{j}^{S_{2}}$, than for the
estimated spectrum corresponding to the eigenvectors
$w_{i}^{S_{1}}\otimes w_{j}^{S_{2}}$. First, 
we present experimental and theoretical results for the estimated eigenvectors and eigenvalues of 
Erd\H{o}s-R\'enyi random networks.

\subsubsection{Experimental results for eigenvectors estimation}
It can be immediately seen that $w_1^{S_1}\otimes w_1^{S_2}$
coincides with an eigenvector of $L_{S_1\otimes S_2}$, where
$w_1^{S_1}$ and $w_1^{S_2}$ are the eigenvectors of $L_{S_1}$ and
$L_{S_2}$, respectively, that correspond to the eigenvalue $0$.
Indeed, since it is well-known that $w_1^{S_1}=1_{S_1}$,
$w_1^{S_2}=1_{S_2}$, $D_{S_1}1_{S_1}= A_{S_1}1_{S_1}$ and
$D_{S_2}1_{S_2}= A_{S_2}1_{S_2}$ we obtain that
\begin{eqnarray*}
L_{S_1\otimes S_2}(w_1^{S_1}\otimes w_1^{S_2})&=& (D_{S_1}\otimes
D_{S_2}-A_{S_1}\otimes A_{S_2})(1_{S_1}\otimes 1_{S_2})\\
&=&D_{S_1}1_{S_1}\otimes D_{S_2}1_{S_2}-A_{S_1}1_{S_1}\otimes
A_{S_2}1_{S_2}=0.
\end{eqnarray*}

We can similarly show that 
$\mathcal{L}_{S_1}\cdot D_{S_1}^{\frac{1}{2}}1_{S_1}=0$ and
$\mathcal{L}_{S_2}\cdot D_{S_2}^{\frac{1}{2}}1_{S_2}=0$. Indeed, we have that
\begin{eqnarray*}
\mathcal{L}_{S_1}\cdot D_{S_1}^{\frac{1}{2}}1_{S_1}&=& (I_{S_1}-D_{S_1}^{-\frac{1}{2}}A_{S_1}D_{S_1}^{-\frac{1}{2}})(D_{S_1}^{\frac{1}{2}}1_{S_1})\\
&=&D_{S_1}^{\frac{1}{2}}1_{S_1}-D_{S_1}^{-\frac{1}{2}}A_{S_1}1_{S_1}\\
&=&D_{S_1}^{\frac{1}{2}}1_{S_1}-D_{S_1}^{-\frac{1}{2}}D_{S_1}1_{S_1}=0.
\end{eqnarray*}
Therefore, for $v_1^{S_1}= D_{S_1}^{\frac{1}{2}}1_{S_1}$ and
$v_1^{S_2}=D_{S_2}^{\frac{1}{2}}1_{S_2}$ it does not hold that
$v_1^{S_1}\otimes v_1^{S_2}$ is an eigenvector of $L_{S_1\otimes
S_2}$. Nevertheless, we omit the examination of the coefficient
correlations that correspond to the vectors $w_1^{S_1}\otimes
w_1^{S_2}$ and $v_1^{S_1}\otimes v_1^{S_2}$  in the following
experimental setup, as we can explicitly calculate the first
eigenvector of $L_{S_1\otimes S_2}$. Moreover, we can not claim in the
general case (for example, when the graphs $S_1$ and $S_2$ are not
regular) that any other approximation vector $w_i^{S_1}\otimes
w_j^{S_2}$ or $v_i^{S_1}\otimes v_j^{S_2}$ coincides with the actual
eigenvector of $L_{S_1\otimes S_2}$.

The first set of experiments was performed for the eigenvectors
comparison of two proposed approximations on the sparse graphs,
that is, we repeat the same experiment as in
\cite{sayama2016estimation} where two Erd\H{o}s-R\'enyi random
networks have 50 vertices (100 edges) and 30 vertices (90 edges),
respectively. It can be easily seen that the edge densities of these
graphs are around 10\%. In Figure~\ref{fig:er_eigenvectors_30x50}
(left panel), one can see the smoothed probability density functions
of vector correlation coefficients between the mentioned vectors
drawn from five independent numerical results. Using the mentioned
parameters on the estimated Laplacian eigenvectors
$w_{i}^{S_{1}}\otimes w_{j}^{S_{2}}$, the correlation coefficients
are above 0.8 in most of the cases, while the peaks are achieved
above 0.9 (green solid lines). For the same graphs, the correlation
coefficients of the eigenvectors $v_{i}^{S_{1}}\otimes
v_{j}^{S_{2}}$ are above 0.7 in most of the cases, while the peaks
are achieved between 0.8 and 0.9 (blue solid lines). Furthermore, it
can be seen in Figure~\ref{fig:er_eigenvectors_30x50} that the
correlation coefficients for the eigenvectors $v_{i}^{S_{1}}\otimes
v_{j}^{S_{2}}$ increase and their graphs shrink to the right (toward
the value of 1) when the edge density level increases (middle and
right panels show the graphics for the edge density levels of 30\%
and 65\%, respectively).

It can be noticed that the correlation coefficients corresponding to the eigenvectors $v_{i}^{S_{1}}\otimes v_{j}^{S_{2}}$ and $w_{i}^{S_{1}}\otimes w_{j}^{S_{2}}$ are symmetrically distributed around the peak and their smoothed probability density functions of vector correlation coefficients look like a probability density function of the normal distribution. Indeed, according to the Pearson's chi-squared test (as a test of goodness of fit) we obtain that most of the correlation coefficients corresponding to the eigenvectors $v_{i}^{S_{1}}\otimes v_{j}^{S_{2}}$ and $w_{i}^{S_{1}}\otimes w_{j}^{S_{2}}$ belong to a fitted normal distribution for the $p$-value of 0.05. When the edge density levels are 10\% for both graphs, 1380 out of 1499 correlation coefficients corresponding to the eigenvectors $v_{i}^{S_{1}}\otimes v_{j}^{S_{2}}$ belong to a fitted normal distribution. For the edge density levels of 30\% and 65\%, 1471 and 1496 out of 1499 correlation coefficients belong to a fitted normal distribution, respectively. On the other hand, when the edge density levels are 10\% for both graphs, 1488 out of 1499 correlation coefficients corresponding to the eigenvectors $w_{i}^{S_{1}}\otimes w_{j}^{S_{2}}$ belong to a fitted normal distribution. For the edge density levels of 30\% and 65\%, 1476 and 1486 out of 1499 correlation coefficients belong to a fitted normal distribution, respectively. Moreover, a similar conclusion can be reached for both vector products when the Erd\H{o}s-R\'enyi random networks have 50 and 100 vertices, as well as 100 and 200 vertices.

\begin{figure}
\centering
\includegraphics[width=\textwidth]{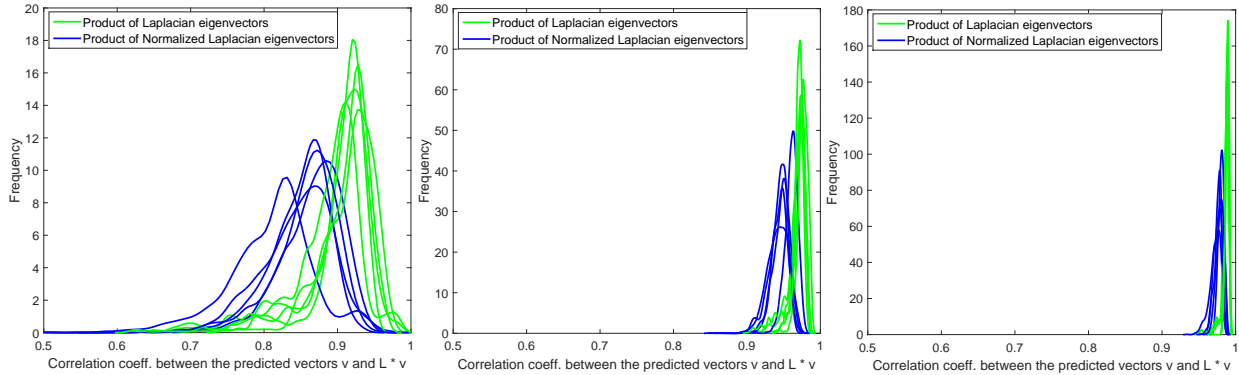} 
\caption{Smoothed probability density functions of vector
correlation coefficients between $w_{i}^{S_{1}}\otimes
w_{j}^{S_{2}}$ and $L_{S_{1}\otimes S_{2}}(w_{i}^{S_{1}}\otimes
w_{j}^{S_{2}})$ are represented by using a green solid line, while
between $v_{i}^{S_{1}}\otimes v_{j}^{S_{2}}$ and $L_{S_{1}\otimes
S_{2}}(v_{i}^{S_{1}}\otimes v_{j}^{S_{2}})$ are represented using a
blue solid line. Probability density functions are drawn for each of
the edge density level 10\%, 30\% and 65\%, respectively, for the
Erd\H{o}s-R\'enyi random graphs with 50 and 30 vertices.}
\label{fig:er_eigenvectors_30x50}
\end{figure}

\subsubsection{Theoretical results for eigenvectors estimation}

Given the performed experiments it can be noticed that some of the
values of correlation coefficients that correspond to the
approximation vectors $w_i^{S_1}\otimes w_j^{S_2}$ are mutually
equal. Indeed, we can explicitly determine the values of correlation
coefficients related to the vectors $w_1^{S_1}\otimes w_j^{S_2}$ and
$w_i^{S_1}\otimes w_1^{S_2}$, for $2\leq j\leq n_2$ and $2\leq i\leq
n_1$, and show that they do not depend on the vectors $w_j^{S_2}$
and $w_i^{S_1}$. In the following text, we prove that the
correlation coefficients related to the vectors $w_1^{S_1}\otimes
w_j^{S_2}$ and $w_i^{S_1}\otimes w_1^{S_2}$ only depend on the
vertex degrees of $S_1$ and $S_2$, respectively.

\begin{theorem}
\label{thm:coeff cor calculation}
  The correlation coefficients $r_{1,j}$ ($2\leq j\leq n_2$) corresponding to the vectors $L_{S_1\otimes S_2}(1_{S_1}\otimes w_j^{S_2})$ and  $1_{S_1}\otimes w_j^{S_2}$ are equal to
  $$
  \frac{\frac{d_1^{S_1}+\cdots+d_{n_1}^{S_1}}{n_1}}{\sqrt{\frac{{d_1^{S_1
}}^2+\cdots+{d_{n_1}^{S_1}}^2}{n_1}}}.
  $$
\end{theorem}
\begin{proof}
Using the fact that $D_{S_1}1_{S_1}=A_{S_1}1_{S_1}=[d_1^{S_1},\ldots,
d_{n_1}^{S_1}]^T$, we show that the vectors $L_{S_1\otimes
S_2}(w_1^{S_1}\otimes w_j^{S_2})$ and $[d_1^{S_1},\ldots,
d_{n_1}^{S_1}]^T\otimes w_j^{S_2}$ are colinear
\begin{eqnarray}
\nonumber
L_{S_1\otimes S_2}(1_{S_1}\otimes w_j^{S_2})&=& (D_{S_1}\otimes D_{S_2}-A_{S_1}\otimes A_{S_2}) (1_{S_1}\otimes w_j^{S_2})\\ \nonumber 
&=& (D_{S_1}1_{S_1})\otimes (D_{S_2}w_j^{S_2})-
(A_{S_1}1_{S_1})\otimes (A_{S_2}w_j^{S_2}) \\ \nonumber &=&
[d_1^{S_1},\ldots, d_{n_1}^{S_1}]^T \otimes (D_{S_2}-A_{S_2})
w_j^{S_2} \\ 
&=&\mu_j^{S_2}[d_1^{S_1},\ldots,
d_{n_1}^{S_1}]^T\otimes w_j^{S_2}.\label{eq:colinearity}\\ \nonumber
\end{eqnarray}

According to (\ref{eq:colinearity}) we have the following chain of equalities
\begin{eqnarray*}
r_{1,j}&=&\frac{\langle L_{S_1\otimes S_2}(1_{S_1}\otimes
w_j^{S_2}), 1_{S_1}\otimes w_j^{S_2}
\rangle}{\|\mu_j^{S_2}[d_1^{S_1},\ldots ,d_{n_1}^{S_1}]^T\otimes w_j^{S_2}\|\cdot \| 1_{S_1}\otimes w_j^{S_2}\|}\\
&=&\frac{(\mu_j^{S_2}[d_1^{S_1},\ldots, d_{n_1}^{S_1}]\otimes
{w_j^{S_2}}^ T)\cdot (1_{S_1}\otimes
w_j^{S_2})}{\mu_j^{S_2}\|[d_1^{S_1},\ldots, d_{n_1}^{S_1}]\|\cdot
\|1_{S_1}\|\cdot \|w_j^{S_2}\|^2}\\
&=&\frac{(\mu_j^{S_2}[d_1^{S_1},\ldots, d_{n_1}^{S_1}]
1_{S_1})\otimes
\|w_j^{S_2}\|^2}{\mu_j^{S_2}\sqrt{n_1}\|[d_1^{S_1},\ldots, d_{n_1}^{S_1}]\|\cdot \|w_j^{S_2}\|^2}\\
&=&\frac{\frac{d_1^{S_1}+\cdots+d_{n_1}^{S_1}}{n_1}}{\sqrt{\frac{{d_1^{S_1
}}^2+\cdots+{d_{n_1}^{S_1}}^2}{n_1}}}.
\end{eqnarray*}
\end{proof}

We see that $r_{1,j}=1$ if and only if the arithmetic mean of the
vertex degrees of $S_1$ is equal to the root mean square of the same
elements and it is well-known that it is true if and only if $S_1$
is a regular graph. On the other hand, the values of $r_{1,j}$ can
be very low in the cases where there is a large gap between the
lowest and highest vertex degrees in the graphic sequence of $S_1$,
$1\leq d_1^{S_1}\leq \ldots \leq d_{n_1}^{S_1}\leq n_1-1$. For
example, considering the complete bipartite graph $S_1=K_{1,n_1-1}$
and calculating the arithmetic mean and the root mean square of the
vertex degrees which are $\frac{2n_1-2}{n_1}$ and $\sqrt{n_1-1}$,
respectively, we can deduce that
$r_{1,j}=\frac{2\sqrt{n_1-1}}{n_1}\rightarrow 0$, when $n_1\rightarrow
\infty$. However, if the sizes of the partition sets in a bipartite
graph become more equal (tend to $n_1/2$) then the coefficient
$r_{1,j}\rightarrow 1$ (for the illustration we can take
$S_1=K_{2,n_1-2}$ and obtain that
$r_{1,j}=\frac{2\sqrt{2n_1-4}}{n_1}>\frac{2\sqrt{n_1-1}}{n_1}$). In
addition, it can be noticed that the correlation coefficients do not
decrease with the increase in the number of different vertex degrees
in $S_1$. Indeed, if we consider the graph $S_1$ with an even order
$n_1=2k+2$ and the graphic sequence
$1,2,\ldots,k,k+1,k+1,k+2,\ldots,2k+1$, it can be determined that
the arithmetic mean and the root mean square are $k+1$ and
$\sqrt{\frac{(2k+1)(4k+3)+3(k+1)}{6}}$, respectively. Therefore, in
this case we obtain high correlation coefficients
$r_{1,j}=\frac{1}{\sqrt{\frac{4}{3}-\frac{1}{2(k+1)}+\frac{1}{6(k+1)^2}}}\rightarrow
\frac{\sqrt{3}}{2}$, $k\rightarrow \infty$. 

However, since we have
obtained correlation coefficients $r_{1,j}$ using a certain number
of synthetic networks produced by the Erd\H{o}s-R\'{e}nyi model, in the
following text we theoretically discuss about the expected values of
the correlation coefficients $r_{1,j}$, using the following auxiliary result.

\begin{prop}[\cite{brockwell} pp. 211]
\label{Brockwell} Suppose that $X_n$ is $AN(\mu, c_n^2\Sigma)$ where
$\Sigma$ is a symmetric nonnegative definite matrix and
$c_n\rightarrow 0$ as $n\rightarrow \infty$. If $g(X) = (g_1 (X),
\ldots ,g_m(X))'$ is a mapping from $R^k$ into $R^m$ such that each
$g_i$ is continuously differentiable in a neighborhood of $\mu$, and
if $D\Sigma D'$ has all of its diagonal elements non-zero, where $D$
is the $m \times k$ matrix $[(\frac{\partial g_i}{\partial
x_j})(\mu)]$, then
$$
g(X_n) {\mbox\ is\ } AN (g(\mu), c_n^2 D\Sigma D').
$$
\end{prop}

\begin{theorem}
If $S_1$ is Erd\H{o}s-R\'{e}nyi graph model, then the expected value of
the correlation coefficient $r_{1,j}$ corresponding to the vectors
$L_{S_1\otimes S_2}(1_{S_1}\otimes w_j^{S_2})$ and  $1_{S_1}\otimes
w_j^{S_2}$ tends to
\begin{eqnarray}
\label{eq:coeffcorr} \sqrt{\frac{(n_1-1)p}{1-p+(n_1-1)p}},
\end{eqnarray}
as $n_1\rightarrow \infty$. \label{thm:corrcoeff_expected}
\end{theorem}
\begin{proof}
Since the distribution of the degree of any particular vertex of the
Erd\H{o}s-R\'{e}nyi graph $S_1=G(n_1,p)$ is binomial, that is
$P(d_i^{S_1}=k)={{n_1-1}\choose k}p^k (1-p)^{n_1-k-1}$, and the fact
that the expected value of any vertex degree is equal to the
expected value of the arithmetic mean of degrees
$Y_1=\frac{d_1^{S_1}+\ldots+d_{n_1}^{S_1}}{n_1}$, we conclude that
 $E(Y_1)=(n_1-1)p$ 
. Furthermore, as $n_1\rightarrow \infty$, according to the central
limit theorem $Y_1$ has asymptotic normal distribution
$AN(\mu_1,\frac{\sigma_1^2}{n_1})$, where $\mu_1=E(Y_1)$ and
$\sigma_1^2=D(d_i^{S_1})=(n_1-1)p(1-p)$ ($E$ and $D$ are usual
notation for expected value and dispersion, respectively).
Similarly, as ${d_i^{S_1}}^2$, $1\leq i\leq n_1$, have the same
distribution, we deduce that $Y_2= \frac{{d_1^{S_1
}}^2+\cdots+{d_{n_1}^{S_1}}^2}{n_1}$ has asymptotic normal
distribution $AN(\mu_2,\frac{\sigma_2^2}{n_1})$, where
$\mu_2=E(Y_2)=E({d_i^{S_1}}^2)$ and $\sigma_2^2=D({d_i^{S_1}}^2)$.
On the other hand, given that
$E({d_i^{S_1}}^2)=D(d_i^{S_1})+E({d_i^{S_1}})^2$ , we have that
$\mu_2=(n_1-1)p(1-p)+(n_1-1)^2p^2$. Considering the two dimensional
variable $X_2=[Y_1,Y_2]$ it can be concluded that its asymptotic
normal distribution is $AN([\mu_1,\mu_2]',c_n^2\Sigma)$, where
$\Sigma$ represents a nonnegative definite symmetric matrix. Define
$g(y_1,y_2)=\frac{y_1}{\sqrt{y_2}}$ which is a continuously
differentiable function. Finally, using Proposition \ref{Brockwell}
we conclude that $g(X_2)$ has asymptotic normal distribution
$AN(\mu_3,c_n^2 D\Sigma D')$, where
$\mu_3=g(\mu_1,\mu_2)=\frac{\mu_1}{\sqrt{\mu_2}}$. Therefore,
the expected value of the coefficient correlation
 $r_{1,j}$ is equal to $E(g(X_2))$
 (according to Theorem \ref{thm:coeff cor calculation}) which tends to $\mu_3= \sqrt{\frac{(n_1-1)p}{1-p+(n_1-1)p}}$, as $n_1\rightarrow \infty$.
\end{proof}

If we rewrite (\ref{eq:coeffcorr}) in the form
$\sqrt{\frac{1}{\frac{1-p}{(n_1-1)p}+1}}$, we conclude that the expected
value of $r_{1,j}$ tends to $1$ when the order of the graph
increases, for the fixed edge density $p$. Therefore, we show that
$w_1^{S_1}\otimes w_j^{S_2}$, $2\leq j\leq n_2$, becomes a more stable
approximation for the larger orders of graphs with constant edge
level. Moreover, we report higher coefficient correlations $r_{i,j}$
when the order of graphs are $50$ and $100$, respectively
(see Fig.~\ref{fig:er_eigenvectors_50x100}). The same conclusion can be obtained if the order of graphs are $100$ and $200$, respectively. 
 
 Similarly,
for a given order $n_1$ of the graph $S_1$, by rewriting
(\ref{eq:coeffcorr}) in the form
$\sqrt{\frac{n_1-1}{\frac{1}{p}-1+(n_1-1)}} $, we conclude that
$r_{1,j}\rightarrow 1$, if $p$ tends to $1$ and $r_{1,j}\rightarrow
0$, if $p$ tends to $0$. Notice that we have already obtained a more
general conclusion by performing three types of experiments in which
the correlation coefficients increase in total as long as the edge
density increases (for a fixed $n_1$). In our experimental setup $p$
can not tend to $0$ since we deal with connected graphs. Namely,
a sharp threshold for the connectedness of $S_1$ is $\frac{\ln
n_1}{n_1}$ (more precisely if $p>\frac{(1+\epsilon)\ln n_1}{n_1}$
then the graph $S_1$ will almost surely be connected). Since the
parameters in the experimental setup satisfy the mentioned
condition, we almost surely deal with connected graphs and after
applying the condition we obtain
$r_{1,j}\geq\frac{1-\frac{1}{n_1}}{\frac{1}{(1+\epsilon)\ln
n_1}+(1-\frac{2}{n_1})}$. Therefore, $r_{1,j}\rightarrow 1$, as $n_1
\rightarrow \infty$, which theoretically confirms our experimental
results that the correlation coefficient $r_{1,j}$ grows as long as
the order of the connected Erd\H{o}s-R\'{e}nyi graph grows.

\begin{figure}
\centering
\includegraphics[width=0.5\textwidth]{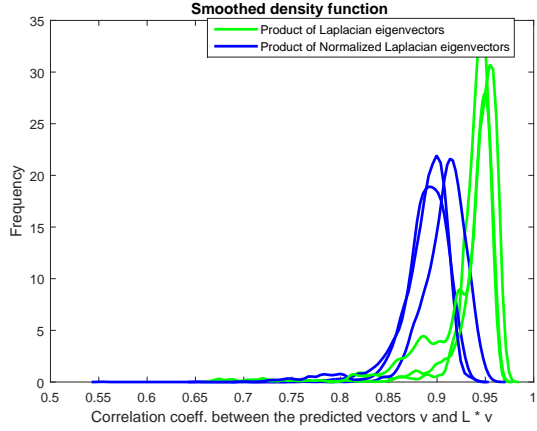} 
\caption{Smoothed probability density functions of vector
correlation coefficients between $w_{i}^{S_{1}}\otimes
w_{j}^{S_{2}}$ and $L_{S_{1}\otimes S_{2}}(w_{i}^{S_{1}}\otimes
w_{j}^{S_{2}})$ are represented by using a solid green line, while
between $v_{i}^{S_{1}}\otimes v_{j}^{S_{2}}$ and $L_{S_{1}\otimes
S_{2}}(v_{i}^{S_{1}}\otimes v_{j}^{S_{2}})$ are represented using a
 solid blue line. Probability density functions are drawn for the
edge density level of 10\%, for the Erd\H{o}s-R\'enyi random graphs
with 50 and 100 vertices.} \label{fig:er_eigenvectors_50x100}
\end{figure}

\bigskip

In the following text, we estimate the correlation coefficients
related to the vectors $v_1^{S_1}\otimes v_j^{S_2}$ and
$v_i^{S_1}\otimes v_1^{S_2}$ in terms of the vertex degrees of $S_1$
and $S_2$, respectively. Moreover, we prove that the expected values of
these coefficients $r'_{1,j}$ can exceed the expected value of $r_{1,j}$ given by
(\ref{eq:coeffcorr}), when $n_1\rightarrow \infty$ and
$n_2\rightarrow \infty$.

\begin{lem}
\label{lem:first} The scalar product of the vectors $L_{S_1\otimes
S_2}(v_1^{S_1}\otimes v_j^{S_2})$ and $v_1^{S_1}\otimes v_j^{S_2}$
is greater than or equal to
  $$
  ({d_1^{S_1}}^2+\cdots+{d_{n_1}^{S_1}}^2)\ \ {v_j^{S_2}}^T L_{S_2}v_j^{S_2},
  $$
  for $1\leq j\leq n_2$. The equality holds true if and only if $S_1$ is
  regular.
\end{lem}
\begin{proof}
Since $v_1^{S_1}=D_{S_1}^{\frac 1 2}1_{S_1}$ we have the following
chain of equalities
\begin{eqnarray}
&&\langle L_{S_1\otimes S_2}(D_{S_1}^{\frac 1 2}1_{S_1}\otimes
v_j^{S_2}), D_{S_1}^{\frac 1 2}1_{S_1}\otimes v_j^{S_2}
\rangle=(D_{S_1}^{\frac 1 2}1_{S_1}\otimes v_j^{S_2})^T
L_{S_1\otimes S_2}(D_{S_1}^{\frac 1 2}1_{S_1}\otimes v_j^{S_2})\nonumber \\
 &=&(D_{S_1}^{\frac 1 2}1_{S_1}\otimes v_j^{S_2})^T
(D_{S_1}\otimes D_{S_2}-A_{S_1}\otimes A_{S_2})(D_{S_1}^{\frac 1
2}1_{S_1}\otimes v_j^{S_2}) \nonumber\\
 &=&(1_{S_1}^T
D_{S_1}^{\frac 1 2}\otimes {v_j^{S_2}}^T) (D_{S_1}\otimes
D_{S_2})(D_{S_1}^{\frac 1 2}1_{S_1}\otimes
v_j^{S_2})-(1_{S_1}^TD_{S_1}^{\frac 1 2}\otimes {v_j^{S_2}}^T)
(A_{S_1}\otimes A_{S_2})(D_{S_1}^{\frac 1 2}1_{S_1}\otimes v_j^{S_2}) \nonumber\\
&=&(1_{S_1}^T D_{S_1}^{\frac 1 2}D_{S_1}D_{S_1}^{\frac 1
2}1_{S_1})\otimes({v_j^{S_2}}^TD_{S_2}v_j^{S_2})-(1_{S_1}^T
D_{S_1}^{\frac 1 2}A_{S_1}D_{S_1}^{\frac 1
2}1_{S_1})\otimes({v_j^{S_2}}^TA_{S_2}v_j^{S_2})\nonumber\\
 &=&(1_{S_1}^T D_{S_1}^21_{S_1})\otimes({v_j^{S_2}}^TD_{S_2}v_j^{S_2})-([{d_1^{S_1}}^{\frac 1 2},\ldots,
 {d_{n_1}^{S_1}}^{\frac 1 2}]A_{S_1}[{d_1^{S_1}}^{\frac 1 2},\ldots, {d_{n_1}^{S_1}}^{\frac 1
 2}]^T)\otimes({v_j^{S_2}}^TA_{S_2}v_j^{S_2}).  \label{eq:last}\\
 \nonumber
\end{eqnarray}
Furthermore, the quadratic forms $1_{S_1}^T D_{S_1}^21_{S_1}$ and
$[{d_1^{S_1}}^{\frac 1 2},\ldots,
 {d_{n_1}^{S_1}}^{\frac 1 2}]A_{S_1}[{d_1^{S_1}}^{\frac 1 2},\ldots, {d_{n_1}^{S_1}}^{\frac 1 2}]^T$
are equal to $\sum_{i=1}^{n_1}{d_i^{S_1}}^2$ and $\sum_{\{i,j\}\in
E(S_1)} 2{d_i^{S_1}}^{\frac 1 2}{d_j^{S_1}}^{\frac 1 2}$,
respectively. According to the inequality of arithmetic and geometric
means it holds that $\sum_{\{i,j\}\in E(S_1)} 2{d_i^{S_1}}^{\frac 1
2}{d_j^{S_1}}^{\frac 1 2}\leq \sum_{\{i,j\}\in E(S_1)}
{d_i}^{S_1}+{d_j}^{S_1}=\sum_{i=1}^{n_1}{d_i^{S_1}}^2$. The equality
holds true if and only if $d_1^{S_1}=\ldots =d_{n_1}^{S_1}$. Finally, we
have that the term (\ref{eq:last}) is greater than or equal to
$\sum_{i=1}^{n_1}{d_i^{S_1}}^2({v_j^{S_2}}^TD_{S_2}v_j^{S_2})-\sum_{i=1}^{n_1}{d_i^{S_1}}^2({v_j^{S_2}}^TA_{S_2}v_j^{S_2})=\sum_{i=1}^{n_1}{d_i^{S_1}}^2({v_j^{S_2}}^TL_{S_2}v_j^{S_2})$.
\end{proof}

\begin{lem}
\label{lem:second} The norm of the vector $L_{S_1\otimes
S_2}(v_1^{S_1}\otimes v_j^{S_2})$ is less than or equal to
  $$
  \sqrt{{d_1^{S_1}}^3+\cdots{d_{n_1}^{S_1}}^3}\ \|L_{S_2}v_j^{S_2}\|,
  $$
  for $1\leq j\leq n_2$. The equality holds true if and only if $S_1$ is
  regular.
\end{lem}
\begin{proof}
We have the following chain of equalities
\begin{eqnarray}
&&L_{S_1\otimes S_2}(D_{S_1}^{\frac 1 2}1_{S_1}\otimes v_j^{S_2})
=(D_{S_1}\otimes D_{S_2}-A_{S_1}\otimes A_{S_2})(D_{S_1}^{\frac 1 2}1_{S_1}\otimes v_j^{S_2})\nonumber \nonumber\\
&=&(D_{S_1}D_{S_1}^{\frac 1 2}1_{S_1})\otimes
(D_{S_2}v_j^{S_2})-(A_{S_1}D_{S_1}^{\frac 1 2}1_{S_1})\otimes
(A_{S_2}v_j^{S_2})\nonumber\\
&=&[{d_1^{S_1}}^{\frac 3 2},\ldots, {d_{n_1}^{S_1}}^{\frac 3
2}]^T\otimes (D_{S_2}v_j^{S_2})-[\sum_{\{1,i\}\in E(S_1)} d_i^{\frac
1 2},\ldots, \sum_{\{n_1,i\}\in E(S_1)} d_i^{\frac 1 2}]^T\otimes
(A_{S_2}v_j^{S_2}).\nonumber\\ \nonumber
\end{eqnarray}
Furthermore, since $u\otimes(Av)=(u\otimes A)v$, where $u^T\in
R^{n_1}$, $v^T\in R^{n_2}$ and $A\in R^{n_2\times n_2}$ it holds
that

\begin{eqnarray}
\label{eq:intermediate2}
& \|L_{S_1\otimes S_2}(D_{S_1}^{\frac 1 2}1_{S_1}\otimes
v_j^{S_2})\| & = \\
&& =\|([{d_1^{S_1}}^{\frac 3 2},\ldots, {d_{n_1}^{S_1}}^{\frac 3
2}]^T\otimes D_{S_2}-[\sum_{\{1,i\}\in E(S_1)} d_i^{\frac 1
2},\ldots, \sum_{\{n_1,i\}\in E(S_1)} d_i^{\frac 1 2}]^T\otimes
A_{S_2})v_j^{S_2}\|. \nonumber
\end{eqnarray}
Now, if we denote $B=[{d_1^{S_1}}^{\frac 3 2},\ldots,
{d_{n_1}^{S_1}}^{\frac 3 2}]^T\otimes D_{S_2}-[\sum_{\{1,i\}\in
E(S_1)} {d_i^{S_1}}^{\frac 1 2},\ldots, \sum_{\{n_1,i\}\in E(S_1)}
{d_i^{S_1}}^{\frac 1 2}]^T\otimes A_{S_2} $ and $A_{S_2}=[a_{i,j}]$,
$1\leq i,j\leq n_2$, then we can easily conclude that
\begin{equation}
\notag B=\begin{bmatrix}
    B_{1}\\
    B_{2}\\
    \vdots\\
    B_{n_{1}}
\end{bmatrix},\quad  B_k = \left \{
  \begin{aligned}
    &{d_k^{S_1}}^{\frac 3 2}d_i^{S_2}, && \text{if}\ i=j \\
    &-\sum_{\{k,l\}\in E(S_1)}
{d_l^{S_1}}^{\frac 1 2}a_{i,j}, && \text{if}\ i\neq j
  \end{aligned} \right.,\quad  1\leq k\leq n_{1}.
\end{equation}

Therefore, we obtain that

\begin{equation}
\label{eq:intermediate}
\begin{array}{lcl}
\|Bv_j^{S_2}\| & = & \|[{d_1^{S_1}}^{\frac 3 2},\ldots,
{d_{n_1}^{S_1}}^{\frac 3 2}]^T \otimes
(D_{S_2}v_j^{S_2})\|+\|[\displaystyle\sum_{\{1,i\}\in E(S_1)} {d_i^{S_1}}^{\frac
1 2},\ldots, \sum_{\{n_1,i\}\in E(S_1)} {d_i^{S_1}}^{\frac 1 2}]^T
\otimes (-A_{S_2}v_j^{S_2})\|\\ 
 & \leq & \|[{d_1^{S_1}}^{\frac 3 2},\ldots, {d_{n_1}^{S_1}}^{\frac 3 2}]^T\| \otimes
\|D_{S_2}v_j^{S_2}\|+\|[\displaystyle\sum_{\{1,i\}\in E(S_1)} {d_i^{S_1}}^{\frac 1 2},\ldots, \sum_{\{n_1,i\}\in E(S_1)} {d_i^{S_1}}^{\frac 1 2}]^T\|\otimes \|(A_{S_2}v_j^{S_2})\|.\\ 
\end{array}
\end{equation}

Furthermore, according to the inequality between the arithmetic mean and root mean square \\ $(\sum_{\{k,i\}\in E(S_1)} {d_i^{S_1}}^{\frac 1 2})^2\leq {d_k^{S_1}} \sum_{\{k,i\}\in E(S_1)} {d_i^{S_1}}$, for $1\leq k\leq n_1$, the following inequalities holds
   
\begin{eqnarray}
\label{eq:intermediate1}
\nonumber
 \|[\sum_{\{1,i\}\in E(S_1)} {d_i^{S_1}}^{\frac 1 2},\ldots,
\sum_{\{n_1,i\}\in E(S_1)} {d_i^{S_1}}^{\frac 1 2}]^T\|^2&\leq&
2\sum_{\{i,j\}\in E(S_1)} {d_i^{S_1}} {d_j^{S_1}}\\ \nonumber
&\leq& \sum_{\{i,j\}\in E(S_1)} {d_i^{S_1}}^2+
{d_j^{S_1}}^2=\sum_{i=1}^{n_1}{d_i^{S_1}}^3=\|[{d_1^{S_1}}^{\frac 3
2},\ldots, {d_{n_1}^{S_1}}^{\frac 3 2}]^T\|^2.\\
\end{eqnarray}

The equality holds true if and only if $d_1^{S_1}=\ldots =d_{n_1}^{S_1}$.
Now, according to the inequalities (\ref{eq:intermediate2}),
(\ref{eq:intermediate}) and (\ref{eq:intermediate1}) we conclude that 

$$
\|L_{S_1\otimes S_2}(D_{S_1}^{\frac 1 2}1_{S_1}\otimes
v_j^{S_2})\|=\|Bv_j^{S_2}\|\leq \sum_{i=1}^{n_1}\sqrt{{d_i^{S_1}}^3}(\|D_{S_2}v_j^{S_2}\|+\|A_{S_2}v_j^{S_2}\|).
$$

From the fact that
$\|L_{S_2}v_j^{S_2}\|=\|(D_{S_2}-A_{S_2})v_j^{S_2}\|=\|D_{S_2}v_j^{S_2}\|+\|A_{S_2}v_j^{S_2}\|$
 we finally have that

$$
\|L_{S_1\otimes S_2}(D_{S_1}^{\frac 1 2}1_{S_1}\otimes
v_j^{S_2})\|\leq \sqrt{\sum_{i=1}^{n_1}{d_i^{S_1}}^3}
\|L_{S_2}v_j^{S_2}\|.
$$

\end{proof}

\begin{theorem}
\label{thm1:coeff cor calculation}
  The correlation coefficients $r'_{1,j}$ ($2\leq j\leq n_2$) corresponding to the vectors $L_{S_1\otimes S_2}(D_{S_1}^{\frac 1 2}1_{S_1}\otimes v_j^{S_2})$ and
  $D_{S_1}^{\frac 1 2}1_{S_1}\otimes v_j^{S_2}$ are greater than or  equal to
  $$
  \frac{{d_1^{S_1}}^2+\cdots+{d_{n_1}^{S_1}}^2}{\sqrt{({d_1^{S_1}}^3+\cdots+{d_{n_1}^{S_1}}^3)({d_1^{S_1}}+\cdots+{d_{n_1}^{S_1}})}}r_j^{S_2},
  $$
  where $r_j^{S_2}$ is the correlation coefficient corresponding to
  the vectors $L_{S_2}v_j^{S_2}$ and $v_j^{S_2}$.
\end{theorem}
\begin{proof}
According to Lemma \ref{lem:first} and Lemma \ref{lem:second} we
obtain that
\begin{eqnarray*}
r'_{1,j}&=&\frac{\langle L_{S_1\otimes S_2}(D_{S_1}^{\frac 1
2}1_{S_1}\otimes v_j^{S_2}), D_{S_1}^{\frac 1 2}1_{S_1}\otimes
v_j^{S_2} \rangle}{\|L_{S_1\otimes
S_2}(D_{S_1}^{\frac 1 2}1_{S_1}\otimes v_j^{S_2})\|\cdot \| D_{S_1}^{\frac 1 2}1_{S_1}\otimes v_j^{S_2}\|}\\
&\geq&\frac{({d_1^{S_1}}^2+\cdots+{d_{n_1}^{S_1}}^2)\ \
{v_j^{S_2}}^T
L_{S_2}v_j^{S_2}}{\sqrt{{d_1^{S_1}}^3+\cdots{d_{n_1}^{S_1}}^3}\
\|L_{S_2}v_j^{S_2}\|\sqrt{d_1^{S_1}+\cdots
d_{n_1}^{S_1}}\|v_j^{S_2}\|}\\
&=& \frac{{d_1^{S_1}}^2+\cdots+{d_{n_1}^{S_1}}^2}
{\sqrt{{d_1^{S_1}}^3+\cdots{d_{n_1}^{S_1}}^3}\
\sqrt{d_1^{S_1}+\cdots d_{n_1}^{S_1}}} \cdot
\frac{{v_j^{S_2}}^TL_{S_2}v_j^{S_2}}{\|L_{S_2}v_j^{S_2}\|\|v_j^{S_2}\|}.
\end{eqnarray*}
\end{proof}

Let us mention that the sum of cubes of vertex degrees of a graph $G$ is known as the forgotten topological index denoted by $F(G)$ \cite{fti}, while the sum of squares of vertex degrees of a graph $G$ represents well known first Zagreb index, denoted by $M_1^1(G)$ \cite{zagrebindex}.
In the following statement we actually prove that the expected value of $\frac{M_1^1(G)}{\sqrt{2mF(t)}}$ is greater than or equal to the expected value of
correlation coefficient $r_{1,j}$ for the random graphs in the asymptotic case. However, it can be shown that $\frac{M_1^1(G)}{\sqrt{2mF(t)}}\geq r_{1,j}$ does not always hold for an arbitrary graph and it would be nice to find the minimum of the function $\frac{M_1^1(G)}{\sqrt{2mF(t)} r_{1,j}}$. This would make a more elegant expression for the upper bound for $F(G)$ than those that can be found in the literature \cite{che2016lower}.

\begin{theorem}
\label{thm:second inequality}
 The asymptotic value of the expected value of the
correlation coefficient $r_{1,j}$ is less than or equal to the
asymptotic value of the expected value of
 $$
  \frac{{d_1^{S_1}}^2+\cdots+{d_{n_1}^{S_1}}^2}{\sqrt{({d_1^{S_1}}^3+\cdots+{d_{n_1}^{S_1}}^3)({d_1^{S_1}}+\cdots+{d_{n_1}^{S_1}})}},
  $$
  as $n_1\rightarrow \infty$.
\end{theorem}
\begin{proof}
According to Theorem \ref{thm:corrcoeff_expected} we have that the
asymptotic value of expected value of the correlation coefficient
$r_{1,j}$ is equal to $ \sqrt{\frac{(n_1-1)p}{1-p+(n_1-1)p}}$, as
$n_1\rightarrow \infty$. On the other hand, as
$P(d_i^{S_1}=k)={{n_1-1}\choose k}p^k (1-p)^{n_1-k-1}$ we have that
$E(Y_1)=n_1(n_1-1)p$ for $Y_1=d_1^{S_1}+\ldots+d_{n_1}^{S_1}$.
Similarly, we can conclude that
$E(Y_2)=n_1((n_1-1)p(1-p)+(n_1-1)^2p^2)$ for
$Y_2={d_1^{S_1}}^2+\cdots+{d_{n_1}^{S_1}}^2$ and
$E(Y_3)=n_1((n_1-1)(n_1-2)(n_1-3)p^3+3p^2(n_1-1)(n_1-2)+(n_1-1)p)$
for $Y_3={d_1^{S_1}}^3+\cdots+{d_{n_1}^{S_1}}^3$. Using Proposition
\ref{Brockwell} we can conduct the similar proof as we do in
Theorem \ref{thm:corrcoeff_expected} and conclude that that asymptotic
value of the expected value of
$\frac{{d_1^{S_1}}^2+\cdots+{d_{n_1}^{S_1}}^2}{\sqrt{({d_1^{S_1}}^3+\cdots+{d_{n_1}^{S_1}}^3)({d_1^{S_1}}+\cdots+{d_{n_1}^{S_1}})}},$
as $n_1\rightarrow \infty$, is equal to
$\frac{E(Y_2)}{\sqrt{E(Y_1)E(Y_3)}}$. It only remains to show that
$$
\frac{n_1((n_1-1)p(1-p)+(n_1-1)^2p^2)}{\sqrt{n_1(n_1-1)p\
n_1((n_1-1)(n_1-2)(n_1-3)p^3+3p^2(n_1-1)(n_1-2)+(n_1-1)p)}}\geq
\sqrt{\frac{(n_1-1)p}{1-p+(n_1-1)p}}.
$$
After a short calculation, the inequality can be reduced to
$$\sqrt{\frac{(1-p+(n_1-1)p)^3}{(n_1-1)(n_1-2)(n_1-3)p^3+3(n_1-1)(n_1-2)p+(n_1-1)p}}\geq 1,$$
which is equivalent to $(n_1-2)p^3-3(n_1-2)p^2+(2n_1-5)p+1\geq 0$.
Furthermore, this can be rewritten in the following way
$$
n_1\geq
2+\frac{2p^3-6p^2+5p-1}{p^3-3p^2+2p}=2+\frac{p-1}{p^3-3p^2+2p}=2+\frac{p-1}{p(p-1)(p-2)}=2-\frac{1}{p(2-p)}.
$$
The arithmetic-geometric mean inequality implies that
$p(2-p)\leq (\frac{p+2-p}{2})^2=1$ and therefore we get $n_1\geq
1\geq 2-\frac{1}{p(2-p)}$, which is obviously true.
\end{proof}

According to Theorem \ref{thm1:coeff cor calculation} and Theorem
\ref{thm:second inequality} we have the following chain of
inequalities
$$
E(r'_{1,j})\geq
E(\frac{{d_1^{S_1}}^2+\cdots+{d_{n_1}^{S_1}}^2}{\sqrt{({d_1^{S_1}}^3+\cdots+{d_{n_1}^{S_1}}^3)({d_1^{S_1}}+\cdots+{d_{n_1}^{S_1}})}}r_j^{S_2})\geq
E(r_{1,j})E(r_j^{S_2}),
$$
as $n_1\rightarrow \infty$. Moreover, we see that the lower bound of $r'_{1,j}$ depends on the degrees of $S_1$ and the correlation coefficient  $r_j^{S_2}$,
while $r_{1,j}$ depends only on the degrees of $S_1$. Therefore, for the higher values $r_{S_2}$
it will be more likely that the expected values of $r'_{1,j}$ is greater than the expected values of $r_{1,j}$.
In fact, if we choose $S_2$ to be the
graph such that $r_j^{S_2}$ is close to $1$ (if $S_2$ is regular then $r_j^{S_2}=1$) we can conclude that $
E(r'_{1,j})\geq E(r_{1,j})$, for every $1\leq j\leq n_2$, as
$n_1\rightarrow \infty$.
\bigskip


\subsubsection{Experimental and theoretical results for eigenvalues estimation}

Furthermore, we show the distributions of percentage errors in
estimated Laplacian spectra of the Kronecker product of graphs
compared to the actual spectrum. The error is calculated over one
hundred independent tests for the Kronecker product of the
Erd\H{o}s-R\'enyi random graphs with 50 and 30 vertices. The errors
for the estimated spectrum corresponding to the eigenvectors
$w_{i}^{S_{1}}\otimes w_{j}^{S_{2}}$ are always drawn on the left
hand side, while the errors for the estimated spectrum corresponding
to the eigenvectors $v_{i}^{S_{1}}\otimes v_{j}^{S_{2}}$ are always
drawn on the right hand side of
Figure~\ref{fig:er_eigenvalues_30x50}. Each row of the figure
corresponds to one of the edge density levels of 10\%, 30\%, and
65\%, respectively. The solid black  curve shows the median, and the
shaded areas show ranges from 5 to 95 percentiles. Notice that when
the edge density increases, the percentage errors become smaller for
both approximations. The characteristic shapes of error
distributions for the estimated spectrum corresponding to the
eigenvectors $w_{i}^{S_{1}}\otimes w_{j}^{S_{2}}$, seen in
Figure~\ref{fig:er_eigenvalues_30x50} (left hand side) have sudden
jumps at the beginning followed by a gradual decrease and they are
fairly consistent across various network density levels that we
tested. There is no a sudden jump at the beginning, for the
estimated spectrum corresponding to the eigenvectors
$v_{i}^{S_{1}}\otimes v_{j}^{S_{2}}$, but there is a small error
widening for the largest eigenvalues. In the case of the estimated
spectrum corresponding to the eigenvectors $w_{i}^{S_{1}}\otimes
w_{j}^{S_{2}}$, the median takes positive values for the
approximately first half of eigenvalues and negative values for the
second half. In the case of the estimated spectrum corresponding to
the eigenvectors $v_{i}^{S_{1}}\otimes v_{j}^{S_{2}}$, the
distribution of percentage errors becomes more stable, that is, the
median is almost a straight line with value 0 for every eigenvalue
(right hand side of Figure~\ref{fig:er_eigenvalues_30x50}). It could be
seen that the error ranges are almost uniformly distributed around
0. \vskip 0.2cm

\begin{figure}
\centering
  \begin{tabular}{@{}c@{}}
     \includegraphics[width=\textwidth]{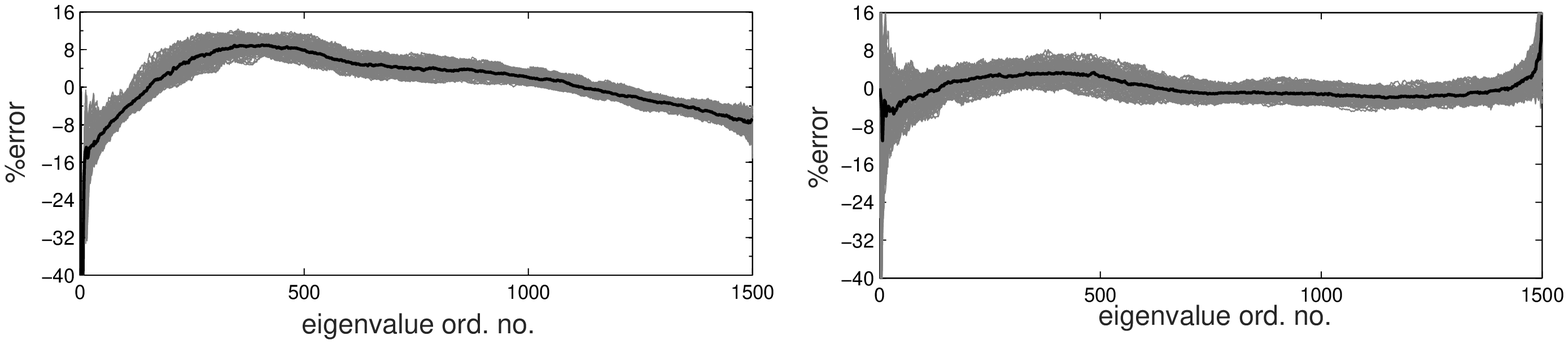}\\
     \includegraphics[width=\textwidth]{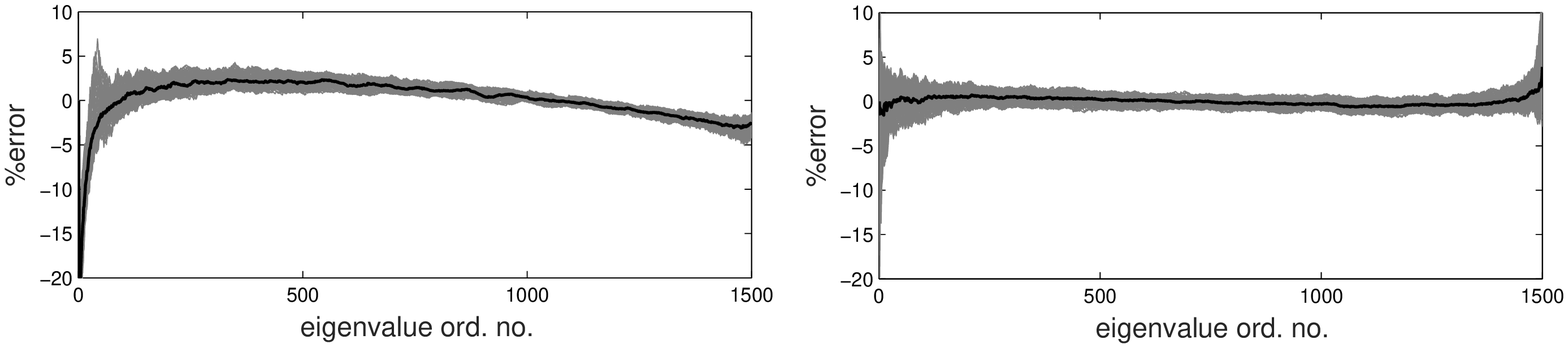}\\
     \includegraphics[width=\textwidth]{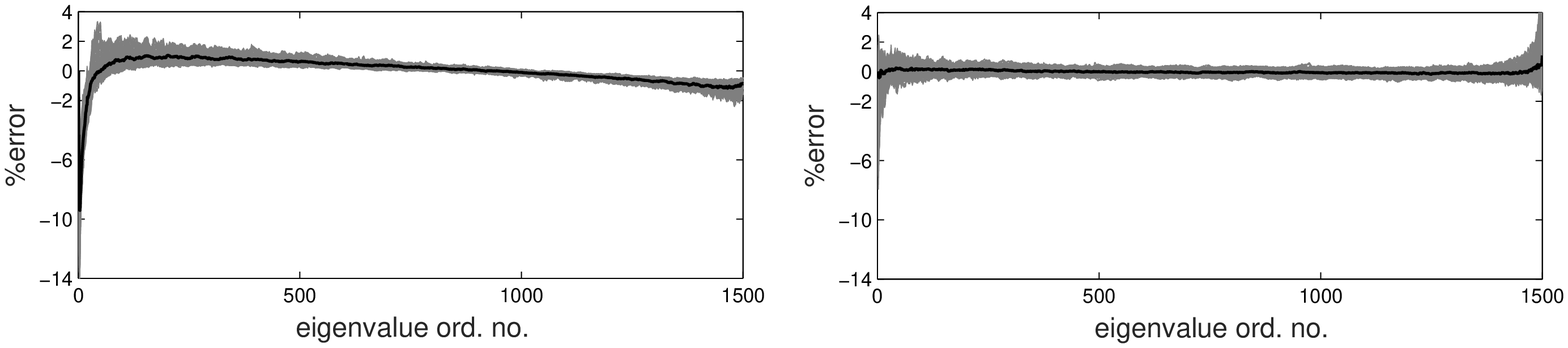}\\
  \end{tabular}
  \caption{Distribution of percentage errors in estimated Laplacian spectra of the Kronecker product of Erd\H{o}s-R\'enyi random graphs (50 and 30 vertices) compared to original ones. \textit{Left hand side} is reserved for the spectrum of the vectors $w_{i}^{S_{1}}\otimes w_{j}^{S_{2}}$ and \textit{right hand side} for the spectrum of the vectors $v_{i}^{S_{1}}\otimes v_{j}^{S_{2}}$. Rows correspond to the edge density levels of 10\%, 30\% and 65\%.}
  \label{fig:er_eigenvalues_30x50}
\end{figure}

Here we give a theoretical explanation of why the estimated
   eigenvalues corresponding to $v_i^{S_1}\otimes v_j^{S_2}$ for the
random graphs become more accurate to the real expected values when
the network grows or the edge density level increases. Conducted experiments show that this approximation produces reasonable estimation of Laplacian spectra with percentage errors confined within a $\pm$10\%, $\pm$5\% and $\pm$2\% range for most eigenvalues when the edge density percentages are 10\%, 30\% and 65\%, respectively.  We use the following statement in order to show a theoretical justification for the above claim.

\begin{theorem}
\label{thm:er_theory} \cite{chung2011spectra} Let $G$ be a random
graph, where $pr(v_{i}\sim v_{j}) = p_{ij}$, and each edge is
independent of each other edge. Let $A$ be the adjacency matrix of
$G$, and $\bar{A} = E(A)$, so $\overline{A}_{ij} = p_{ij}$. Let $D$
be the diagonal matrix with $D_{ii} = deg(v_{i})$, and $\bar{D}=
E(D)$. Let $\bar{\delta}=\bar{\delta}(G)$ be the minimum expected
degree of $G$, and $\mathcal{L}$ the normalized Laplacian matrix for
$G$. For any $\epsilon > 0$, if there exists a constant $k =
k(\epsilon)$ such that  $\bar{\delta} > k\, ln\, n$, then with
probability at least $1-\epsilon$, the $j$-th eigenvalues of
$\mathcal{L}$ and $\bar{\mathcal{L}}$ satisfy

$$|\lambda_{j}(\mathcal{L})-\lambda_{j}(\bar{\mathcal{L}})|\leq 2\sqrt{\frac{3 ln(\frac{4n}{\epsilon})}{\bar{\delta}}}$$

\noindent for all $1\leq j \leq n$, where $\bar{\mathcal{L}} = I -
\bar{D}^{-\frac{1}{2}}\bar{A}\,\bar{D}^{-\frac{1}{2}}$.
\end{theorem}

Let $G_{n,p}$ be a random graph with order $n$ and
probability of creation of an edge $p$. Since in the experiments we
use the factor graphs with the same edge density percentage (denote
these graphs by $G_{n_1,p_1}$ and $G_{n_2,p_2}$), without loss of
generality, we may set $p_{1}=p_{2}=p$ (an identical analysis can be
conducted when $p_{1}\neq p_{2}$). For the expected adjacency
matrices of the random graphs $G_{n_{1},p}$ and $G_{n_{2},p}$ hold
$\bar{A}_{n_{1}} = p(J_{n_{1}} - I_{n_{1}})$ and $\bar{A}_{n_{2}} =
p(J_{n_{2}} - I_{n_{2}})$. By $A_{n_1}$ and $A_{n_2}$ we denote the adjacency matrices
of $G_{n_1,p_1}$ and $G_{n_2,p_2}$. By $\mathcal{L}(G_{n_1,p}\otimes G_{n_2,p})$ we also denote the normalized Laplacian matrix for the graph $G_{n_1,p}\otimes G_{n_2,p}$.

First, we show that $\bar{\delta}=\bar{\delta}(G_{n_1,p}\otimes G_{n_2,p})\sim
n_1n_2$. Notice also that since the sum of each row of the matrix
$\bar{A}_{n_{1}}\otimes \bar{A}_{n_{2}}$ is equal to
$p^{2}(n_{1}-1)(n_{2}-1)$, then it is clear that
$\delta(\bar{G}_{n_1,p}\otimes \bar{G}_{n_2,p}) =
p^{2}(n_{1}-1)(n_{2}-1)$. Let $Z=\min\{d^1_id^2_k\ |\ 1\leq i\leq
n_1\ 1\leq k\leq n_2\}$, where $d_i^1$ and $d_k^2$ are the degrees
of the vertices in $G_{n_1,p_1}$ and $G_{n_2,p}$, respectively.
Therefore, we have that $\bar{\delta}=E(Z)$. According to Jensen's
inequality, it holds that $e^{-t\bar{\delta}}\leq E(e^{-tZ})$, for
any positive real $t$. Furthermore, according to the definition of
$Z$, we have the following chain of relation

\begin{eqnarray}
e^{-t\bar{\delta}}\leq
E(e^{-tZ})=E(e^{-t\min_{i,j} \{d^1_id^2_k\}})= E(\max_{i,j} e^{-t
d^1_id^2_k})\leq \sum_{i,j}E( e^{-t d^1_id^2_k})=n_1n_2E( e^{-t
d^1_id^2_k}),
\label{eq:delta_bar}
\end{eqnarray}
for any $1\leq i\leq n_1$ and $1\leq k\leq n_2$. 

As $n_1, n_2\rightarrow \infty$, according to the central limit
theorem $d^1_i$ and $d^2_k$ have asymptotic normal distribution
$AN(\mu_1,\sigma_1^2)$ and $AN(\mu_2,\sigma_2^2)$, respectively,
where $\mu_1=n_1p$, $\mu_2=n_2p$, $\sigma_1=\sqrt{n_1pq}$ and
$\sigma_2=\sqrt{n_2pq}$. Considering the two dimensional variable
$X=[d^1_i,d^2_k]$ it can be concluded that it has asymptotic normal
distribution and since $g(x,y)=e^{-t xy}$ is a continuously
differentiable function we conclude that $g(X)$ has an asymptotic
normal distribution. Therefore, when $n_1, n_2\rightarrow \infty$,
we have that $E( e^{-t d^1_id^2_k})=\frac{1}{2\pi\sigma_1\sigma_2}
\int_{-\infty}^{\infty}\int_{-\infty}^{\infty} e^{-txy}e^{\frac 1
2(\frac{x-\mu_1}{\sigma_1})}e^{\frac 1
2(\frac{y-\mu_2}{\sigma_2})}dxdy$. After the substitutions
$x\rightarrow \frac{x-\mu_1}{\sigma_1}$, $y\rightarrow
\frac{y-\mu_2}{\sigma_2}$ and certain number of elementary
algebraic transformations we obtain that

\begin{eqnarray*}
E( e^{-t d^1_id^2_k})&=&\frac{e^{-t \mu_1\mu_2}}{2\pi}
\int_{-\infty}^{\infty} e^{-t\mu_1\sigma_2y-\frac{y^2}{2}}
e^{\frac{t^2\sigma_1^2(\mu_2+\sigma_2y)^2}{2}}
\int_{-\infty}^{\infty}
e^{\frac{(x+t\sigma_1(\mu_2+\sigma_2y))^2}{2}}dx dy\\
&=& \frac{e^{-t \mu_1\mu_2}}{\sqrt{2\pi}} \int_{-\infty}^{\infty}
e^{-t\mu_1\sigma_2y-\frac{y^2}{2}+\frac{t^2\sigma_1^2(\mu_2+\sigma_2y)^2}{2}}dy.
\end{eqnarray*}

The last integral can be rewritten in the following form
$\frac{1}{\sqrt{2\pi}}e^{-t \mu_1\mu_2} e^{\frac{{t^2
\sigma_1^2\mu_2^2}}{2}} \int_{-\infty}^{\infty} e^{\frac{-y^2
A-2yB}{2}}$, where $A=1-t^2\sigma_1^2\sigma_2^2$  and
$B=-t\mu_1\sigma_2+t^2\sigma_1^2\mu_2\sigma_2$. Finally, we have
that
\begin{eqnarray*}
E( e^{-t d^1_id^2_k})&=&\frac{1}{\sqrt{2\pi}}e^{-t
\mu_1\mu_2+\frac{{t^2
\sigma_1^2\mu_2^2}}{2}}e^{\frac{B^2}{2A}}\int_{-\infty}^{\infty}
e^{\frac{-(\sqrt{A}(y-\frac{B}{A}))^2}{2}}dy=\frac{1}{\sqrt{2\pi}}e^{-t
\mu_1\mu_2+\frac{{t^2
\sigma_1^2\mu_2^2}}{2}+\frac{B^2}{2A}}\frac{\sqrt{2\pi}}{\sqrt{A}}\\
&=& \frac{e^{-t \mu_1\mu_2+\frac{{t^2
\sigma_1^2\mu_2^2}}{2}+\frac{(-t\mu_1\sigma_2+t^2\sigma_1^2\mu_2\sigma_2)^2}{2(1-t^2\sigma_1^2\sigma_2^2)}}}{\sqrt{1-t^2\sigma_1^2\sigma_2^2}}\\
&=&\frac{e^{\frac{-2t
\mu_1\mu_2+(\mu_1^2\sigma_2^2+\mu_2^2\sigma_1^2)t^2}{2(1-t^2\sigma_1^2\sigma_2^2)}}}{\sqrt{1-t^2\sigma_1^2\sigma_2^2}}.
\end{eqnarray*}

According to (\ref{eq:delta_bar}) we obtain
$$
\bar{\delta}\geq -\frac{ln(n_1n_2)}{t}-\frac{-2
\mu_1\mu_2+(\mu_1^2\sigma_2^2+\mu_2^2\sigma_1^2)t}{2(1-t^2\sigma_1^2\sigma_2^2)}+\frac{ln(1-t^2\sigma_1^2\sigma_2^2)}{2t},
$$
for every $t>0$. Now, if we set $t=\frac{1}{\mu_2\sigma_1}$, we can
easily obtain that the leading summand of the right hand side of the
above inequality is $\mu_1\mu_2$, hence we further conclude that
$\bar{\delta}=\Omega (n_1n_2)$, when $n_1,n_2\rightarrow \infty$.

\bigskip

Let $Spectrum(\bar{A}_{n_{1}})$, $Spectrum(\bar{A}_{n_{2}})$,
$Spectrum(\bar{A}_{n_{1}}\otimes \bar{A}_{n_{2}})$ and
$Spectrum(\mathcal{L}(\bar{A}_{n_{1}}\otimes \bar{A}_{n_{2}}))$ be
the multisets of the eigenvalues of the matrices $\bar{A}_{n_{1}}$,
$\bar{A}_{n_{2}}$, $\bar{A}_{n_{1}}\otimes \bar{A}_{n_{2}}$ and
$\mathcal{L}(\bar{A}_{n_{1}}\otimes \bar{A}_{n_{2}})$, respectively.
In order to calculate $Spectrum(\mathcal{L}(\bar{A}_{n_{1}}\otimes
\bar{A}_{n_{2}}))$, we need to determine the diagonal matrix
$D(\bar{A}_{n_{1}}\otimes \bar{A}_{n_{2}})$,
$Spectrum(\bar{A}_{n_{1}})$, $Spectrum(\bar{A}_{n_{2}})$ and
$Spectrum(\bar{A}_{n_{1}}\otimes \bar{A}_{n_{2}})$, but for the sake
of simplicity, these steps are skipped. So, the normalized Laplacian
spectrum of the expected adjacency matrix of the Kronecker product
of two random graphs consists of

\begin{equation}
{1\,\,\, n_{2}-1\,\,\,\, n_{1}-1\,\,\,\, (n_{1}-1)(n_{2}-1)} \choose
{0\,\,\,\,\,\, \frac{n_{2}}{n_{2}-1}\,\,\,\,
\frac{n_{1}}{n_{1}-1}\,\,\,\, 1-\frac{1}{(n_{1}-1)(n_{2}-1)}}
\label{eq:estimated_spectrum}
\end{equation}
where the second row represents the eigenvalues, while the first row
represents the corresponding algebraic multiplicities.

Since $\overline{\delta}=\Omega(n_1n_2)\gg ln(n_{1}n_{2})$, we can
apply Theorem~\ref{thm:er_theory} by putting $\epsilon =
\frac{1}{\sqrt{n_{1}n_{2}}}$ and obtain

\begin{equation}
\label{eq:normalized_spectrum}
|\lambda_{j}(\mathcal{L}(G_{n_{1},p}\otimes
G_{n_{2},p}))-\lambda_{j}(\mathcal{L}(\bar{A}_{n_{1}}\otimes
\bar{A}_{n_{2}}))| \leq 2\sqrt{\frac{3ln4+
\frac{9ln(n_{1}n_{2})}{2}}{ n_{1}n_{2}}} = o(1),
\end{equation}
with probability greater than or equal to $1 -
\frac{1}{\sqrt{n_{1}n_{2}}} = 1 - o(1)$.

\smallskip

In the following, we estimate the difference between $d_i^1d_k^2$
and $\bar{\delta}$ by using Chebyshev's inequality, i.e.
$Pr(|d_i^1d_k^2-\bar{\delta}|<\epsilon\sigma(d_i^1d_k^2))\geq
1-\frac{1}{\epsilon^2}$, for any real $\epsilon>0$. Since $d_i^1$
and $d_k^2$ are independent, we have that
$\sigma^2(d_i^1d_k^2)=\mu_1\sigma_2+\mu_2\sigma_1+\sigma_1\sigma_2$.
Therefore, for $\epsilon=\sqrt[4]{n_1n_2}$ it can be concluded that
$$
|d_i^1d_k^2-\bar{\delta}|<\sqrt{\sqrt{n_1n_2}(\mu_1\sigma_2+\mu_2\sigma_1+\sigma_1\sigma_2)}
$$
with probability greater than or equal to $1 -
\frac{1}{\sqrt{n_{1}n_{2}}} = 1 - o(1)$. Furthermore, since $0\leq\lambda_{j}(\mathcal{L}(G_{n_{1},p}\otimes
G_{n_{2},p}))\leq 2$ and $
\lambda_{j}(L(\bar{A}_{n_{1}}\otimes \bar{A}_{n_{2}}))=\bar{\delta}
\lambda_{j}(\mathcal{L}(\bar{A}_{n_{1}}\otimes
\bar{A}_{n_{2}}))=\bar{\delta}O(1)=\Omega(n_1n_2)O(1)$, which
follows from the formula $L =
D^{\frac{1}{2}}\mathcal{L}D^{\frac{1}{2}}$ and the property that the
graph $\bar{G}_{n_{1},p}\otimes \bar{G}_{n_{2},p}$ is regular, it
holds that
\begin{equation}
\label{eq:first summand}
\frac{|d_i^1d_k^2-\bar{\delta}|\lambda_{j}(\mathcal{L}(G_{n_{1},p}\otimes
G_{n_{2},p}))}{\lambda_{j}(L(\bar{A}_{n_{1}}\otimes
\bar{A}_{n_{2}}))}<\frac{\sqrt{\sqrt{n_1n_2}(\mu_1\sigma_2+\mu_2\sigma_1+\sigma_1\sigma_2)}}{\Omega(n_1n_2)O(1)}=o(1).
\end{equation}

\smallskip

On the other hand, by multiplying both hand sides of the
inequality~\eqref{eq:normalized_spectrum} with $\bar{\delta}$ and
dividing by $\lambda_{j}(L(\bar{A}_{n_{1}}\otimes
\bar{A}_{n_{2}}))$, we obtain
\begin{eqnarray}
\label{eq:second summand}
\frac{
|\bar{\delta}\lambda_{j}(\mathcal{L}(G_{n_{1},p}\otimes
G_{n_{2},p})) - \lambda_{j}(L(\bar{A}_{n_{1}}\otimes
\bar{A}_{n_{2}}))|}{\lambda_{j}(L(\bar{A}_{n_{1}}\otimes
\bar{A}_{n_{2}}))}\leq \frac{\bar{\delta}\,
o(1)}{\bar{\delta}O(1)}=o(1).
\end{eqnarray}

By adding the inequalities (\ref{eq:first summand}) and
(\ref{eq:second summand}), we finally conclude that
\begin{eqnarray}
\notag
&&\frac{|d_i^1d_k^2\lambda_{j}(\mathcal{L}(G_{n_{1},p}\otimes
G_{n_{2},p})) - \lambda_{j}(L(\bar{A}_{n_{1}}\otimes
\bar{A}_{n_{2}}))|}{\lambda_{j}(L(\bar{A}_{n_{1}}\otimes
\bar{A}_{n_{2}}))}\\
&=&\frac{|d_i^1d_k^2\lambda_{j}(\mathcal{L}(G_{n_{1},p}\otimes
G_{n_{2},p}))-\bar{\delta}\lambda_{j}(\mathcal{L}(G_{n_{1},p}\otimes
G_{n_{2},p}))+\notag
\bar{\delta}\lambda_{j}(\mathcal{L}(G_{n_{1},p}\otimes G_{n_{2},p}))
- \lambda_{j}(L(\bar{A}_{n_{1}}\otimes
\bar{A}_{n_{2}}))|}{\lambda_{j}(L(\bar{A}_{n_{1}}\otimes
\bar{A}_{n_{2}}))}\\
&\leq&\frac{|d_i^1d_k^2\lambda_{j}(\mathcal{L}(G_{n_{1},p}\otimes
G_{n_{2},p}))-\bar{\delta}\lambda_{j}(\mathcal{L}(G_{n_{1},p}\otimes
G_{n_{2},p}))|+|\notag
\bar{\delta}\lambda_{j}(\mathcal{L}(G_{n_{1},p}\otimes G_{n_{2},p}))
- \lambda_{j}(L(\bar{A}_{n_{1}}\otimes
\bar{A}_{n_{2}}))|}{\lambda_{j}(L(\bar{A}_{n_{1}}\otimes
\bar{A}_{n_{2}}))}=o(1).
\end{eqnarray}

\smallskip



In the previous formula we show that percentage error between the estimated spectra and the spectra of Laplacian of expected Kronecker random graph tends to zero, when $n_1$ and $n_2$ tend to infinity, while 
in the performed experiments we calculate the
percentage error between the estimated and actual spectra (estimated spectra is given by (\ref{eq:novel_spectrum2})). Therefore, in the rest of the section we give an asymptotic estimate of the percentage error between the  estimated spectra and the mean of the eigenvalues of Laplacian matrix.

Indeed, some empirical evidence indicate that the mean of the empirical distribution of the eigenvalues of the Laplacian matrix
of $G(n,p)$ is centered around $np$ (see \cite{Olivier-Leveque}). Similarly, if we denote mean of the empirical distribution of the eigenvalues of the Laplacian matrix
of $G(n_1,p)\otimes G(n_1,p)$ by $\bar{\lambda}$, we can conclude that $\bar{\lambda}\sim n_1n_2$ and therefore

\begin{equation}
\frac{|d_i^1d_k^2\lambda_j(\mathcal{L}(G_{n_{1},p}\otimes
G_{n_{2},p})) - \bar{\lambda}(L(A_{n_{1}}\otimes
A_{n_{2}}))|}{\bar{\lambda}(L(A_{n_{1}}\otimes
A_{n_{2}}))}  =
o(1). \label{eq:tend}
\end{equation}

Therefore, in
that case we conclude that the formula~\eqref{eq:tend} represents
the percentage error of the estimated spectrum
$d_i^1d_k\lambda_j(\mathcal{L}(G_{n_{1},p}\otimes
G_{n_{2},p}))$ from (\ref{eq:novel_spectrum2}), which tends to 0
when the order of the graph or its edge density tends to infinity.

\bigskip

\textit{Watts-Strogatz random graphs} \vskip 0.2cm

Similarly, we apply the same experiments when two
graphs are Watts-Strogatz graphs. By examining the spectral
properties of the Kronecker product of graphs that are
Watts-Strogatz graphs, we notice that the situation is a bit
different since even when the graphs are sparse (edge density level
is 10\%), the smoothed probability density functions of the vector
correlation coefficients are shrank toward the value of 1, for both
approximations. For the same density, peaks for both approximations
are located in the interval $(0.9, 1)$. When the edge density level
is 30\% and more,  extremely high values of  the correlation coefficients
become more noticeable. In Figure~\ref{fig:ws_eigenvectors_30x50}
the smoothed probability density functions of vector correlation
coefficients are drawn when two graphs are Watts-Strogatz random
graphs with 50 and 30 vertices. The figure shows correlation
coefficients from five independent numerical results when the edge
density level is set to 10\% (left), 30\% (middle) and 65\% (right).

As in the case of Erd\H{o}s-R\'enyi graphs, the distribution of
percentage errors of the estimated spectrum corresponding to the
eigenvectors $v_{i}^{S_{1}}\otimes v_{j}^{S_{2}}$ is almost
uniformly distributed around 0 for each tested edge density, while
the distribution of percentage errors of the estimated spectrum
corresponding to the eigenvectors $w_{i}^{S_{1}}\otimes
w_{j}^{S_{2}}$ always has a sudden jump at the beginning. In
Figure~\ref{fig:ws_eigenvalues_30x50}, errors for the estimated
spectrum from Subsection~\ref{subsec:sayama_approx} are drawn on the
left side, while errors for the estimated spectrum from
Subsection~\ref{subsec:novel_approx} on the right side are drawn. As in case of the Erd\H{o}s-R\'enyi random graphs, both approximations produced reasonable estimations of Laplacian spectra with percentage errors confined within a $\pm$10\% and less as the edge density percentage becomes higher.

\vskip 0.2cm

\begin{figure}
\centering
\includegraphics[width=\textwidth]{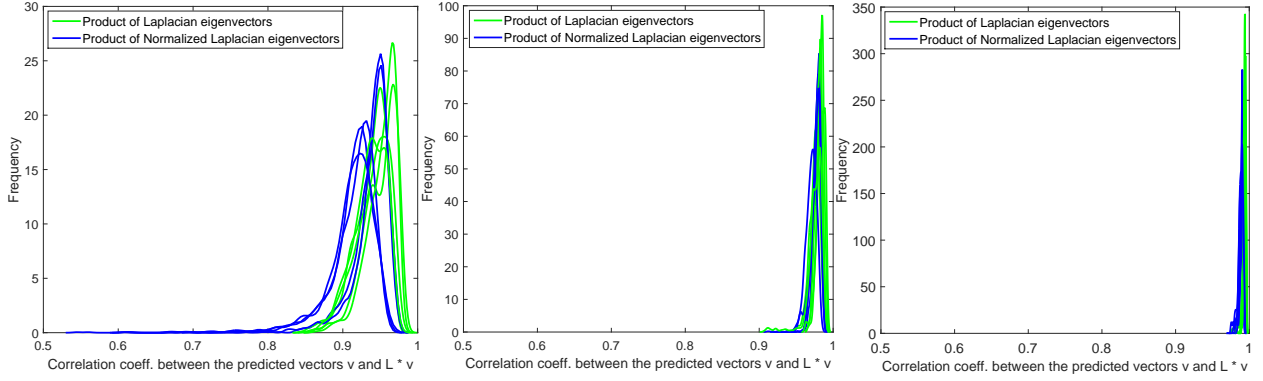}  
\caption{Smoothed probability density functions of vector
correlation coefficients between $w_{i}^{S_{1}}\otimes
w_{j}^{S_{2}}$ and $L_{S_{1}\otimes S_{2}}(w_{i}^{S_{1}}\otimes
w_{j}^{S_{2}})$ are represented using a solid green  line, while
between $v_{i}^{S_{1}}\otimes v_{j}^{S_{2}}$ and $L_{S_{1}\otimes
S_{2}}(v_{i}^{S_{1}}\otimes v_{j}^{S_{2}})$ are represented using a solid
blue  line. Watts-Strogatz random graphs have 50 and 30
vertices. Probability density functions are drawn for each of the
edge density level 10\%, 30\% and 65\%, respectively.}
\label{fig:ws_eigenvectors_30x50}
\end{figure}

\begin{figure}
\centering
  \begin{tabular}{@{}c@{}}
     \includegraphics[width=\textwidth]{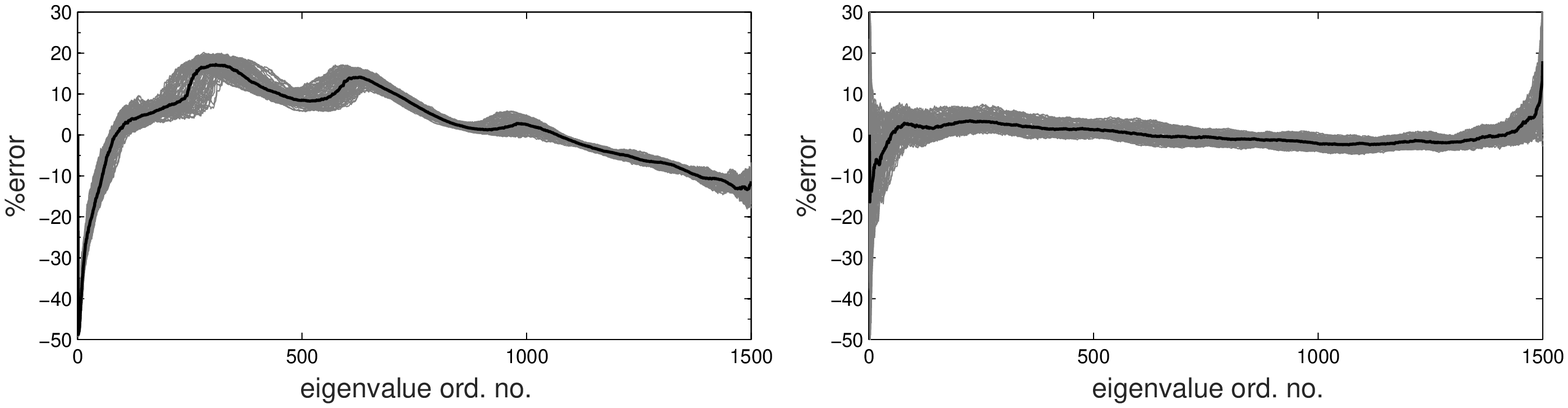}\\
     \includegraphics[width=\textwidth]{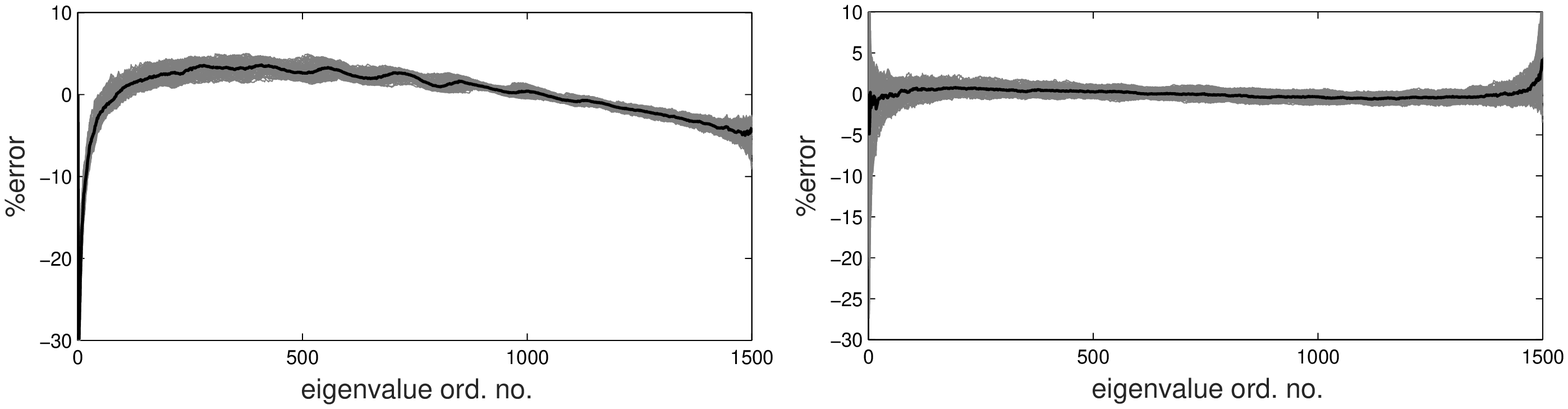}\\
     \includegraphics[width=\textwidth]{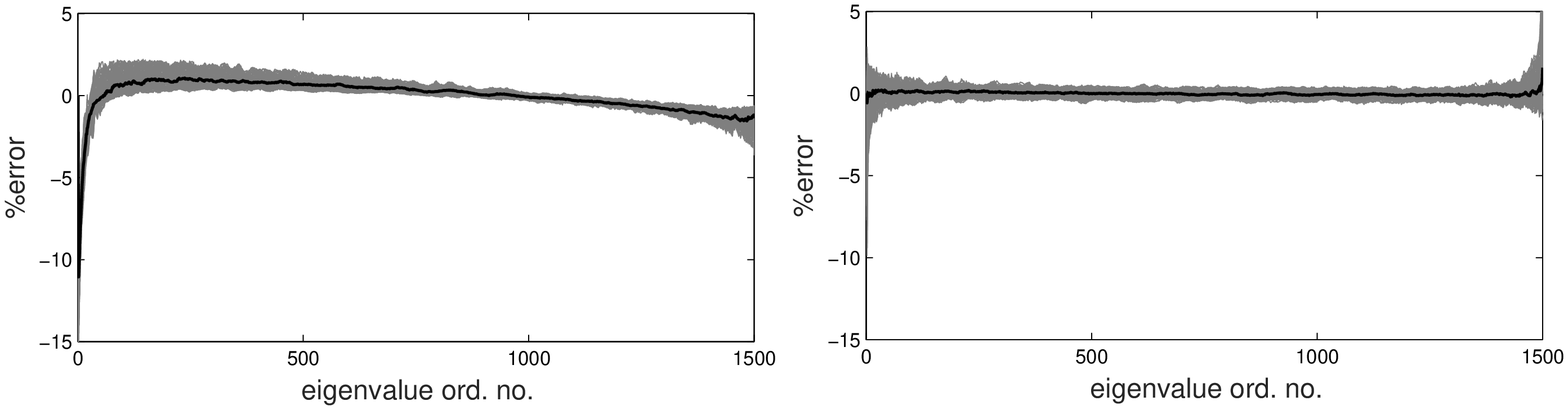}
  \end{tabular}
  \caption{Distribution of percentage errors in estimated Laplacian spectra of the Kronecker product of Watts-Strogatz graphs (50 and 30 vertices) compared to original ones. \textit{Left hand side} is reserved for the spectrum of the eigenvectors $w_{i}^{S_{1}}\otimes w_{j}^{S_{2}}$ and \textit{right hand side} for the spectrum of the eigenvectors $v_{i}^{S_{1}}\otimes v_{j}^{S_{2}}$. Rows correspond to the edge density levels of 10\%, 30\% and 65\%.}
  \label{fig:ws_eigenvalues_30x50}
\end{figure}

\subsection{Barab\'asi-Albert graphs}
\label{subsec:ba_vectors}

In this section we present a behavior of the eigenvectors and
eigenvalues of the Kronecker product of two graphs which are
Barab\'asi-Albert graphs. For this type of graph, the situation is
not significantly different compared to the previous two types
concerning correlation coefficients of the estimated eigenvalues. In
all cases $w_{i}^{S_{1}}\otimes w_{j}^{S_{2}}$ eigenvectors express
better properties, since their correlation coefficients are above
0.9 in most of the cases, while the correlation coefficients of the
$v_{i}^{S_{1}}\otimes v_{j}^{S_{2}}$ eigenvectors are in interval
(0.7, 0.9) most of cases (see
Figure~\ref{fig:ba_eigenvectors_30x50}), for the edge density levels
of 10\%, 30\%, and 65\%.

Also, we notice that the estimated eigenvalues corresponding to the
eigenvectors $w_{i}^{S_{1}}\otimes w_{j}^{S_{2}}$ are more stable
than the eigenvalues corresponding to the eigenvectors
$v_{i}^{S_{1}}\otimes v_{j}^{S_{2}}$. From
Figure~\ref{fig:ba_eigenvalues_30x50} it can be noticed that the
error ranges (which correspond to edge density levels of
10\%, 30\% and 65\%) between the estimated and original spectrum are less for the first approximation (left panels) than for the second one, which are at the same time more distorted (right panels). The
characteristic shape of error distribution for the estimated
spectrum corresponding to the eigenvectors $w_{i}^{S_{1}}\otimes
w_{j}^{S_{2}}$ remains similar as in the previous subsection. This
includes a sudden jump at the beginning followed by a gradual
decrease across various network density levels we tested. Unlike the
previous subsection, the estimated spectrum corresponding to the
eigenvectors $v_{i}^{S_{1}}\otimes v_{j}^{S_{2}}$ has a sudden jump
at the beginning and the error ranges are a little bit higher for
all edge densities (right panels of
Figure~\ref{fig:ba_eigenvalues_30x50}). When the edge density is from 50\% to 65\%, the characteristic shape for the second approximation (right panels) is a bit different than usual. It could be noticed sudden jump in the middle of graphs, but in the same time error narrowing for the largest eigenvalues.

\begin{figure}
\centering
\includegraphics[width=\textwidth]{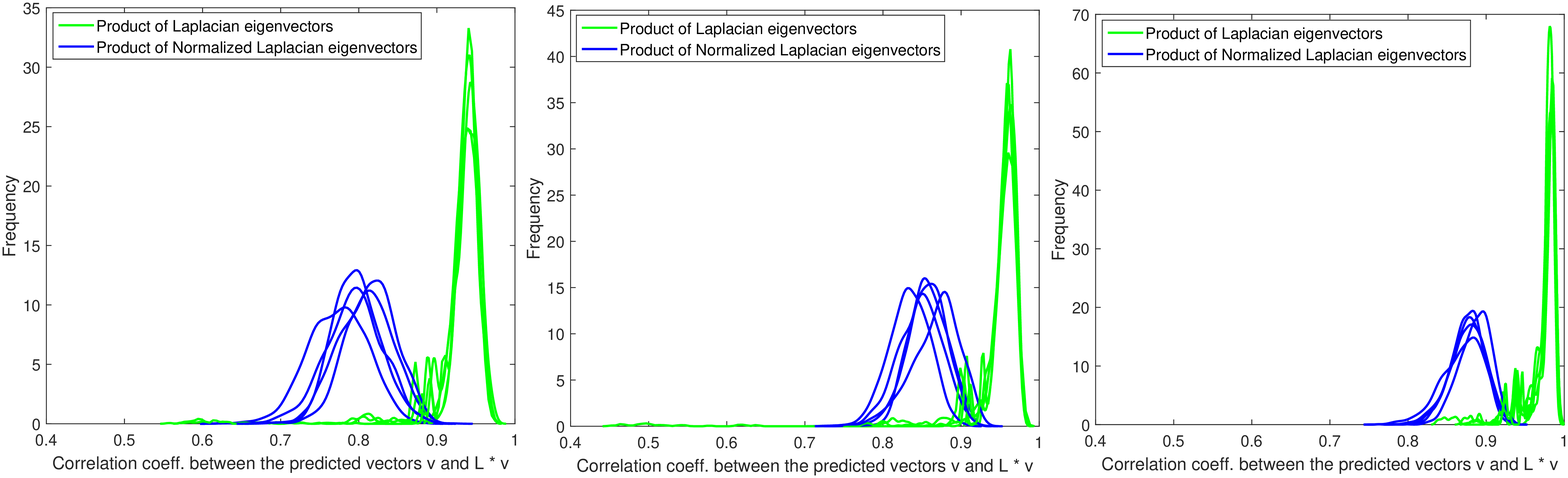} 
\caption{Smoothed probability density functions of vector
correlation coefficients between $w_{i}^{S_{1}}\otimes
w_{j}^{S_{2}}$ and $L_{S_{1}\otimes S_{2}}(w_{i}^{S_{1}}\otimes
w_{j}^{S_{2}})$ are represented using a solid green  line, while
between $v_{i}^{S_{1}}\otimes v_{j}^{S_{2}}$ and $L_{S_{1}\otimes
S_{2}}(v_{i}^{S_{1}}\otimes v_{j}^{S_{2}})$ are represented using a
solid blue  line. Barab\'asi-Albert random graphs have 50 and 30
vertices. Probability density functions are drawn for each of the
edge density 10\%, 30\% and 65\%, respectively.}
\label{fig:ba_eigenvectors_30x50}
\end{figure}

\begin{figure}
\centering
  \begin{tabular}{@{}c@{}}
     \includegraphics[width=\textwidth]{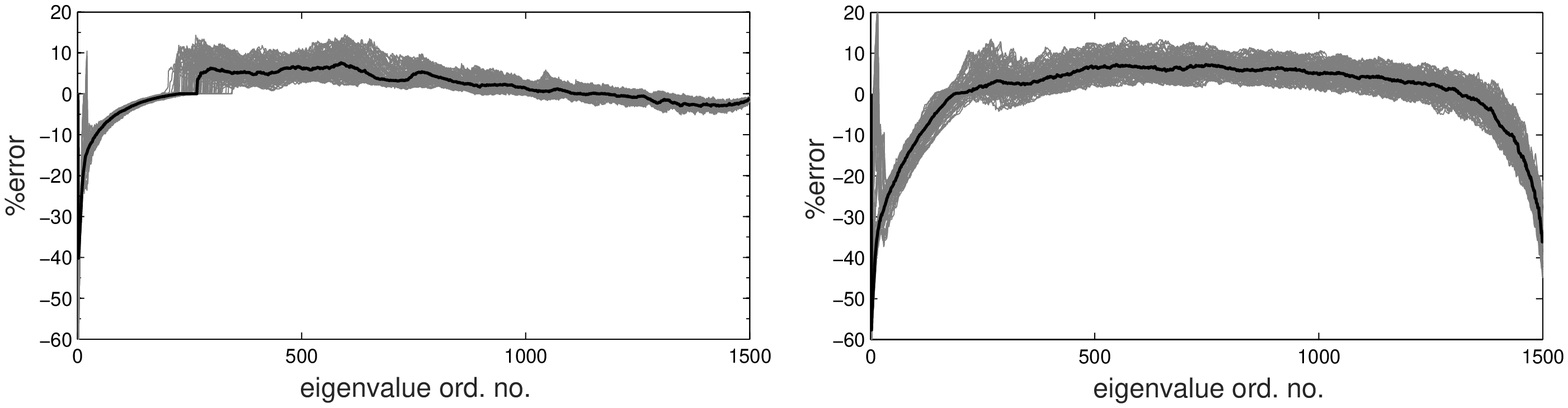}\\
     \includegraphics[width=\textwidth]{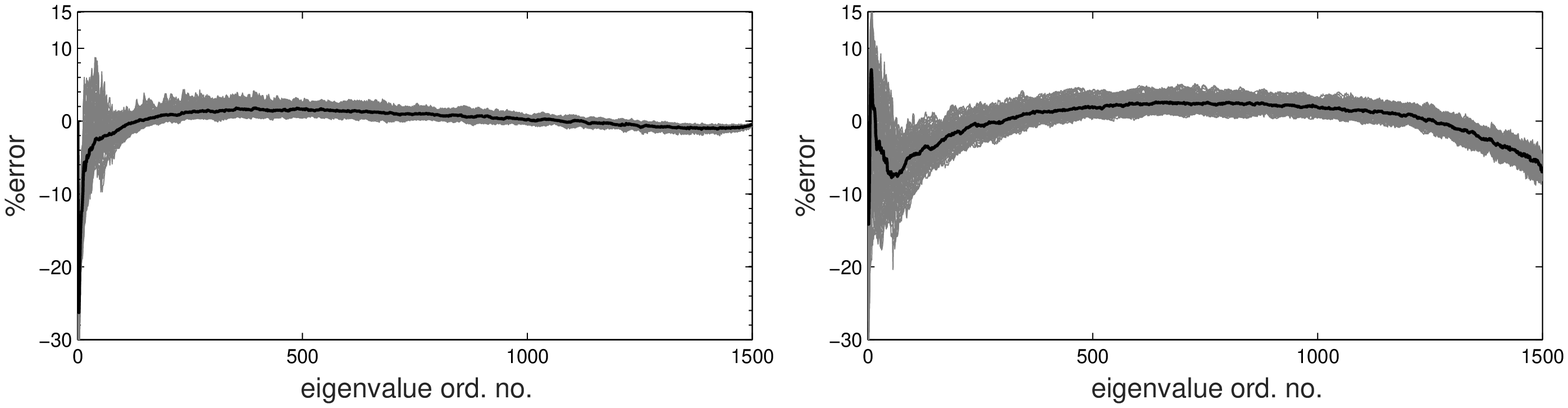}\\
     \includegraphics[width=\textwidth]{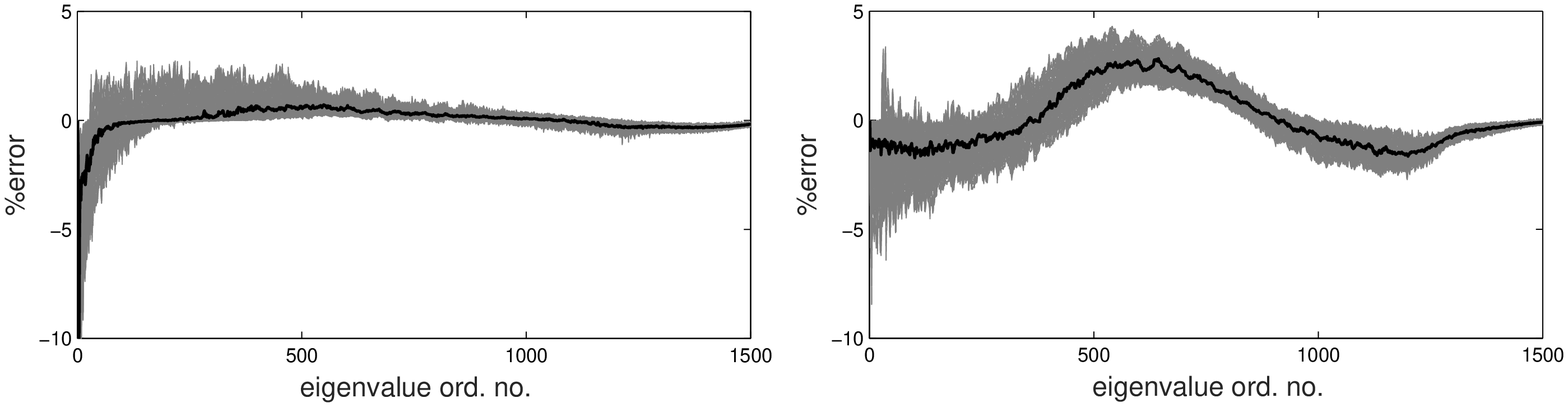}
  \end{tabular}
  \caption{Distribution of percentage errors in estimated Laplacian spectra of the Kronecker product of Barab\'asi-Albert graphs (50 and 30 vertices) compared to original ones. \textit{Left hand side} is reserved for the spectrum of the eigenvectors $w_{i}^{S_{1}}\otimes w_{j}^{S_{2}}$ and \textit{right hand side} for the spectrum of the eigenvectors $v_{i}^{S_{1}}\otimes v_{j}^{S_{2}}$. Rows correspond to the edge density levels of 10\%, 30\%, and 65\%.}
  \label{fig:ba_eigenvalues_30x50}
\end{figure}

\section{Conclusion}
Although the relationships between spectral properties of a product
graph and those of its factor graphs have been known for the
standard products, characterization of Laplacian spectrum and
eigenvectors of the Kronecker product of graphs using the Laplacian
spectra of the factors has remained an open problem to date. In this
work we proposed a novel approximation method for estimating the
Laplacian spectrum and the corresponding eigenvectors of the
Kronecker product of graphs knowing the eigenvalues and eigenvectors
of factor graphs. The estimated eigenvalues and eigenvectors were
compared to the original ones with regard to different types of
random networks and theirs edge density levels. Moreover, the
properties of the novel approximation were compared with the
approximation proposed by Sayama. Although both approximations were
designed using a few mathematically incorrect assumptions, the
obtained estimations of the spectra are very close to the
numerically calculated spectra with percentage errors constrained
within a $\pm$10\% range for most eigenvalues. Here, we give a
theoretical explanation of why the estimated eigenvalues for the
random graphs become more accurate to the real values when the
network grows or the edge density level increases. This explains the
fact that a distribution of percentage errors between estimated and
original spectra becomes almost uniformly distributed around 0. In
this paper we also presented some novel theoretical results related
to the certain correlation coefficients corresponding to the
estimated and original vectors. Here, we provide an exact formula of
how some of these correlation coefficients can be explicitly
calculated, as well as their expected values for some types of
random networks.

As it was mentioned earlier, in this and Sayama's paper, these
approximations have many theoretical limitations, because of the
mathematically incorrect assumptions and there is no rigorous
mathematical explanation of why and how the proposed methods work.
That is why a design of spectral estimation algorithms will be an
important direction of future research, as well as their theoretical
explanations. 
Moreover, it would be very important to see how the estimated eigenvalues and eigenvectors are
suitable for complete spectral decomposition of the graph, where all
eigenvalues and eigenvectors are included to replace original ones.  According to some preliminary results we have already obtained by incorporating these approximations in the GCRF model, a good behaviour of these approximations presented in this paper have been experimentally confirmed too. Moreover, we obtained that the presented estimations can be a good staring point for other applications and further improvements of Laplacian spectrum of the Kronecker product of graphs.


\newpage
\bibliography{references.bib}

\end{document}